\documentclass[12pt]{article}
\usepackage{amssymb,amsmath,graphicx,stmaryrd,mathtools, amsthm, color}
\usepackage{authblk}
\usepackage[left=3cm, right=3cm]{geometry}
\usepackage{hyperref}

\newtheorem{theorem}{Theorem}
\newtheorem{lemma}[theorem]{Lemma}
\newtheorem{proposition}[theorem]{Proposition}

\newtheorem{definition}[theorem]{Definition}
\newtheorem{example}[theorem]{Example}

\newcommand{\figcaption}[1]{\caption{\small #1}}

\theoremstyle{remark}
\newtheorem{remark}{Remark}

\newcommand{\ket}[1]{\left|#1\right\rangle}
\newcommand{\bra}[1]{\left\langle#1\right|}
\newcommand{\ADV}{{\rm{ADV}}}
\newcommand{\wsize}{{\rm{wsize}}}
\newcommand{\Span}{{\rm{span}}}

\title{Span Program for Non-binary Functions}
\author[1]{Salman Beigi} 
\affil[1]{\it \small School of Mathematics, Institute for Research in Fundamental Sciences (IPM), Tehran, Iran}

\author[2]{Leila Taghavi}
\affil[2]{\it \small School of Computer Science, Institute for Research in Fundamental Sciences (IPM), Tehran, Iran}

\date{}

\begin{document}

\maketitle
\begin{abstract}
Span programs characterize the quantum query complexity of \emph{binary} functions $f:\{0,\ldots,\ell\}^n \to \{0,1\}$ up to a constant factor. In this paper we generalize the notion of span programs for functions with \emph{non-binary} input/output alphabets $f: [\ell]^n \to [m]$. 
We show that \emph{non-binary span program} characterizes the quantum query complexity of any such function up to a constant factor. We argue that this non-binary span program is indeed the generalization of its binary counterpart.
We also generalize the notion of span programs for a special class of relations. 

Learning graphs provide another tool for designing quantum query algorithms for binary functions. In this paper, we also generalize this tool for non-binary functions,  and as an application of our non-binary span program show that any non-binary learning graph gives an upper bound on the quantum query complexity. 
 
\end{abstract}

\section{Introduction}
Query complexity of a function is the number of queries to its input bits required to compute it.  In quantum query algorithms queries can be made in superposition, so we sometimes observe speed-up in quantum query algorithms comparing to their classical counterparts, e.g., a quadratic speed-up in Grover's algorithm~\cite{Grover96}. 

The quantum query complexity is characterized by semi-definite programs (SDP). 
The \emph{generalized adversary method}~\cite{Amb02, HLS07} is a certain semi-definite programming optimization problem which gives lower bounds on the quantum query complexity of a function. Surprisingly, the dual of such a SDP, called the \emph{dual adversary bound}, gives an upper bound on the quantum query complexity~\cite{LMRS10}. Thus  the generalized adversary bound characterizes the quantum query complexity of all functions up to a constant factor.

Although each feasible point of the dual adversary bound results in a quantum query algorithm, finding good such feasible solution is a hard problem in general since the size of this SDP is usually so huge (exponential in the size of the input of the function) that makes it intractable. Span programs~\cite{Rei09} and learning graphs~\cite{Bel14} are two methods for finding such solutions which result in quantum query algorithms and upper bounds on the quantum query complexity if the underlying function is binary. Indeed, similar to the generalized adversary bound, span programs characterize the quantum query complexity up to a constant factor~\cite{Rei09}, while learning graphs only provide upper bounds on it. 
Comparing to finding feasible solutions for the dual adversary SDP,
these methods are sometimes more effective for designing quantum query algorithms (see e.g.,~\cite{Bel11,LMM11}). Nevertheless, span programs work only for binary functions $f: \{0,1\}^n\rightarrow \{0,1\}$ and the method of learning graphs work only for functions with binary output. 
To overcome this difficulty Ito and Jeffery~\cite{ItoJeffery15} put forward a definition of span program that works for functions with non-binary input alphabet, and showed that such span programs characterize the quantum query complexity of functions of the form $f:[\ell]^n\to \{0,1\}$. However, in their work, the output of the function is still binary.

\paragraph{Our results:} In this paper we generalize the methods of span programs and learning graphs for non-binary functions $f: D_f\subseteq [\ell]^n\rightarrow [m]$, i.e., functions with non-binary input and/or non-binary output alphabets. 
Building on the work of It and Jeffery~\cite{ItoJeffery15}, we introduce a notion of non-binary span program and prove that it characterizes the quantum query complexity of any function $f: D_f\subseteq [\ell]^n\rightarrow [m]$ up to a constant factor. Thus non-binary span program can also be used to design quantum query algorithms. 

We introduce yet another notion of non-binary span program which we call non-binary span program \emph{with orthogonal inputs}. The latter is indeed a special case of our non-binary span program, yet since it is more intuitive and easier to understand we study this special case separately. We show that non-binary span programs with orthogonal inputs characterize the quantum query complexity of any function $f: D_f\subseteq [\ell]^n\rightarrow [m]$ up to a factor of order $\sqrt{\ell-1}$. 
 In particular, non-binary span programs with orthogonal inputs work equally well as their binary counterparts when the input alphabet is binary while the output alphabet is arbitrary. 

Our second contribution is to generalize the notion of learning graphs for non-binary functions in Section~\ref{sec:NBLG}. We show that non-binary learning graphs give an upper bound on the quantum query complexity of non-binary functions. Our proof of this result about non-binary learning graphs is based on non-binary span programs.

The reader may suggest that a non-binary function $f:[\ell]^n\rightarrow [m]$ can be thought of as a collection of $\log(m)$ binary functions, so there is no point in generalizing span programs to the case of non-binary output. Indeed, each letter of the input alphabet is a $\log(\ell)$-bit string and similarly for the output bits. Thus $f$ can be thought  of as $\log(m)$ binary functions and to estimate the quantum query complexity of $f$ one can design $\log(m)$ binary span programs, one for each of these functions. While this approach does give an upper bound on the quantum query complexity of $f$, its effectiveness is questionable:
\begin{itemize}
\item First, to get a bound on the quantum query complexity with this approach we need to design several (binary) span programs, which at least for some functions does not seem effective.
\item Second, it is not clear how  tight such a bound would be. Our non-binary span program gives a bound that is tight up to a constant factor. 
\item Third, and more importantly, the above suggestion usually kills off the whole point of span program that is an intuitive way of designing feasible points of the dual adversary SDP; thinking of a non-binary function as a collection of binary ones usually destroys the intuition behind the definition of that function.  For example, think of the \emph{Max} function which outputs the maximum of a collection of numbers (and will be worked out later in this paper). It is much easier to think of the maximum of some numbers comparing to its individual bits.
\end{itemize}

\section{Preliminaries}
In this section we review the notions of generalized adversary bound and span program. 
We first fix some notations.
Throughout the paper we use Dirac's ket-bra notation, e.g., $\ket v$ is a complex (column) vector whose conjugate transpose is $\bra v$; Moreover, $\langle v| w\rangle$ is the inner product of vectors $\ket v, \ket w$. 
For a matrix $A$, we denote its $(i,j)$-th entry by $A\llbracket i,j\rrbracket$. The Hadamard (entry-wise) product of two matrices $A$ and $B$ is denoted by $A\circ B$. The Hermitian conjugate of matrix $A$ is denoted by  $A^\dagger$ which is obtained from $A$ by taking the transpose and then taking the complex conjugate of each entry.
 $A\succeq 0$ means that $A$ is a positive semi-definite matrix. $\|A\|$ is the operator norm of the matrix $A$, i.e., the maximum singular value of $A$. We also use
$$[\ell] = \{0, \ldots, \ell-1\}.$$ 
Note that this is unlike the convention that $[\ell]$ denotes $\{1, \dots, \ell\}$ since we would like $0$ to be a symbol in our alphabets.   
Finally, the Kronecker delta symbol $\delta_{a,b}$ is equal to $1$ if $a$ and $b$ are equal, and is $0$ otherwise.

We emphasize that all functions considered throughout the paper are assumed to be partial functions with domain $D_f\subseteq [\ell]^n$ and output set $[m]$, so that  $f:D_f\rightarrow [m]$.
We say that $f$ has a non-binary input alphabet if $\ell>2$, and has a non-binary output alphabet if $m>2$.

In this paper we deal with the problem of computing a function $f:D_f\to[m]$ in the query model, in which case its input $x= (x_1, \dots, x_n)\in D_f\subseteq  [\ell]^n$ is given via queries to its coordinates. 
In the classical setting a query is of the form ``what is the value of the $j$-th coordinate of the input?'' The answer to this query would be the value of $x_j$. Queries are made \emph{adaptively}, i.e., after each query the algorithm decides what to do next based on the values of all the previously queried indices and queries the next index if needed. 

In the quantum setting, a query can be made in superposition. Such queries to an input $x$ are modeled by a unitary operator $O_x$ as follows:
\begin{equation*}
O_x|j,p\rangle=|j,(x_j+p) \mod \ell \rangle.
\end{equation*}
Here the first register contains the coordinate index $j$, and the second register saves the value of $x_j$ in a reversible manner. Thus a quantum query algorithm for computing $f(x)$ is a sequence of unitaries some of which are $O_x$ and the others are independent of $x$ (but can depend on $f$ itself). At the end a measurement determines the outcome of the algorithm. We say that the algorithm computes $f$, if for every $x\in D_f\subseteq [\ell]^n$ the outcome of the algorithm equals $f(x)$ with probability at least $2/3$.
The complexity of such an algorithm is the maximum number of queries, i.e., the number of $O_x$'s in the sequence of unitaries. We denote by $Q(f)$ the \emph{quantum query complexity} of $f$, namely, the minimum query complexity among all quantum algorithms that compute $f$.

\subsection{ The Adversary Bound and its Dual}\label{ADVbound and dual}
The \emph{generalized adversary bound} introduced by H{\o}yer, Lee and \v{S}palek~\cite{HLS07}, based on the work of Abmainis~\cite{Amb02}, gives a lower bound on the quantum query complexity of a function $f:D_f\to [m]$ with $D_f\subseteq [\ell]^n$. 
A \emph{symmetric} matrix $\Gamma\in \mathbb{R}^{|D_f|\times |D_f|}$ whose rows and columns are indexed by elements of the domain of $f$,
is called an \emph{adversary matrix} for $f$ if  
\begin{equation*}
\Gamma\llbracket x,y\rrbracket=0, \qquad \forall x, y\in D_f ~\text{ s.t.}  ~f(x)=f(y).
\end{equation*}
Observe that arranging elements of $D_f$ based on the values $f(x)$, an adversary matrix takes the form of an $m\times m$ block matrix all of whose diagonal blocks are zero. 
 We also let $\Delta_j \in \mathbb R^{|D_f|\times |D_f|}$ with 
 $$\Delta_j\llbracket x,y\rrbracket=1-\delta_{x_j,y_j}.$$
Now the general adversary bound (hereafter, adversary bound) is defined by
\begin{equation}\label{eq:ADV+}
\ADV^\pm(f)=\max_{\Gamma\neq 0} \,\frac{\|  \Gamma \| }{\max_{1\leq j\leq n} \|  \Gamma \circ \Delta_j\| },
\end{equation}
where the maximum is taken over all non-zero (symmetric) adversary matrices $\Gamma$. As mentioned before the adversary bound is a lower bound on the quantum query complexity: $\ADV^\pm(f)\leq Q(f)$.

The adversary bound can be rewritten in the form
\begin{subequations}\label{SDP:ADV}
\begin{align}
\max &\quad \| \Gamma\|  \\
\text{subject to} &\quad \| \Gamma\circ \Delta_j\| \leq 1 \qquad \forall j, ~1\leq j\leq n, \label{eq:ADV,const1}\\
 &\quad\Gamma \text{: adversary matrix,}
\end{align}
\end{subequations}
which is a semi-definite program (SDP). The dual of this SDP is called the \emph{dual adversary bound} and takes the form:
\begin{subequations}\label{SDP:GDADV}
\begin{align}
\min & \quad\max_{z\in D_f}~\sum_{j=1}^n (X_j+Y_j)\llbracket z,z\rrbracket \\
\text{subject to} &\quad \sum_{j: x_j\neq y_j}(X_j-Y_j)\llbracket x,y\rrbracket =1 \qquad \forall x, y \in D_f \textnormal{ s.t. } f(x)\neq f(y),\\
&\quad  X_j,Y_j\succeq 0 \qquad  \forall j, ~1\leq j\leq n.
\end{align}
\end{subequations}
Here $X_j,Y_j\in \mathbb R^{|D_f|\times |D_f|}$ are positive semi-definite matrices whose rows and columns are labeled by elements of $D_f$.

The optimal value of the dual adversary bound~\eqref{SDP:GDADV} equals the following SDP up to a  factor of at most 2~\cite{LMRS10}.
\begin{subequations}\label{eq:dual-SDP}
\begin{align} 
\min &\quad \max_{x\in D_f}\; \max\bigg\{\sum_{j=1}^n \big\| |u_{xj}\rangle \big\| ^2, \sum_{j=1}^n \big\| |v_{xj}\rangle \big\| ^2\bigg\}\\
\text{subject to} & \quad \sum_{j: x_j\neq y_j} \langle u_{xj}|v_{yj}\rangle=1-\delta_{f(x),f(y)} \qquad \forall x,y\in D_f,
\end{align}
\end{subequations}
with vectors $|u_{xj}\rangle,|v_{xj}\rangle $'s being the variables. In the sequel by the dual adversary SDP we sometimes mean~\eqref{eq:dual-SDP} since it is essentially equivalent to~\eqref{SDP:GDADV}.

The constraint~\eqref{eq:ADV,const1} is equivalent to
\begin{subequations}
\begin{align}
 \Gamma\circ \Delta_j\preceq I, \qquad \forall j,\label{eq:extraC0}\\
-I\preceq \Gamma\circ \Delta_j, \qquad \forall j.\label{eq:extraC}
\end{align}
\end{subequations}
These two constraints correspond to two sets of variables $X_j$'s and $Y_j$'s in the dual SDP~\eqref{SDP:GDADV}. 
Now when the output set of $f$ is binary, i.e., $m=2$, for any adversary matrix $\Gamma$, the spectrum of $\Gamma\circ\Delta_j$ is symmetric with respect to $0$ (because $\Gamma\circ\Delta_j$ is a $2\times 2$ symmetric block matrix with zero blocks on the diagonal). In this case~\eqref{eq:extraC0} and~\eqref{eq:extraC} are equivalent to each other, and one of them can be dropped. Indeed, when $m=2$, the constraint~\eqref{eq:ADV,const1} can be replaced with 
\begin{equation}
\Gamma\circ \Delta_j\preceq I, \qquad \forall j.
\end{equation}
Then the variables in the dual SDP associated to the other constraints~\eqref{eq:extraC} can be dropped and set to be zero, i.e., when $m=2$ with no loss of generality in~\eqref{SDP:GDADV} we may put $Y_j=0$. Similarly in this case in~\eqref{eq:dual-SDP} the vectors $\ket{u_{xj}}$ and $\ket{v_{xj}}$ can be taken to be equal.

The adversary bound, being a lower bound on the quantum query complexity $Q(f)$, is also an upper bound on $Q(f)$ up to a constant factor. More precisely, any feasible point of the dual adversary bound~\eqref{SDP:GDADV} (and of~\eqref{eq:dual-SDP}) results in a quantum query algorithm for $f$ whose query complexity equals the objective value up to a constant factor. 
This fact was first proved by Reichardt~\cite{Rei09} for functions with binary output and then generalized for all functions by Lee et al. \cite{LMRSS11}. As a result, to determine $Q(f)$ it is enough to compute $\ADV^\pm(f)$, and to design quantum query algorithms it is enough to find feasible solutions to~\eqref{SDP:GDADV} or~\eqref{eq:dual-SDP}.

Finding desirable feasible points of~\eqref{SDP:GDADV} or~\eqref{eq:dual-SDP} is a hard job in general because usually the size of these SDPs is so huge. Span programs and learning graphs are two intuitive methods for finding such solutions which have already resulted in several quantum query algorithms. In the rest of this section we describe span programs, and defer learning graphs for Section~\ref{sec:NBLG}.

\subsection{Span Program} \label{sec:SP}
Span program is a model of computation that was first introduced by Karchmer and Wigderson~\cite{KW93}. This tool has been used for designing quantum algorithms by \v{S}palek and Reichardt~\cite{RS12}. Later,  Reichardt~\cite{Rei09} used span programs for designing quantum query algorithms.

 A \textit{span program} $P$ evaluating a binary function ($\ell=m=2$) function $f:D_f\rightarrow\{0,1\}$ with $D_f\subseteq \{0,1\}^n$ consists of
\begin{itemize}
\item a finite-dimensional inner product space $\mathbb{C}^d$,
\item a non-zero target vector $|t\rangle\in \mathbb{C}^d$,
\item and input vector sets $I_{j,q}\subseteq\mathbb{C}^d$, for all $1\leq j\leq n$ and $q\in \{0,1\}$.
\end{itemize}
Given this data, we define $I\subseteq\mathbb{C}^d$ by
$$I=\bigcup_{j=1}^n \bigcup_{q\in \{0,1\}} I_{j,q}.$$ 
Also, for $x\in D_f$ we define the set of \emph{available vectors} by
$$I(x)=\bigcup_{j=1}^n I_{j,x_j}.$$
Thus vectors in $I_{j,q}$ would be available when the $j$-th coordinate of the input $x$ is equal to $q$.
Moreover, we let $A\in \mathbb{C}^{|I|\times d}$ be the matrix consisting of column vectors in $I$ (see Figure~\ref{fig:Aspanprogram}).

We say that the span program $P$ \emph{evaluates} the function $f$ whenever $f(x)=1$ if and only if $|t\rangle \in \text{span} I(x)$:
\begin{equation*}
 |t\rangle \in \text{span}\, I(x) \quad\text{ iff } \quad x\in f^{-1}(1).
\end{equation*}
Then for every $x\in f^{-1}(x)$ there exists $\ket{w_x}\in \mathbb C^{|I|}$, called a \emph{positive witness} for $x$, such that the coordinates of $\ket{w_x}$ associated to unavailable vectors are zero, and $A\ket{w_x} = \ket t$. Indeed, $\ket{w_x}$ witnesses the fact that $\ket t$ belongs to $\mathrm{span}\, I(x)$. 

Also, if $x\in f^{-1}(0)$ then $\ket t\notin  \mathrm{span} \,I(x)$. Therefore, there exists a vector $\ket{\bar w_x} \in \mathbb C^{d}$ called a \emph{negative witness} such that $\langle \bar w_x | v\rangle=0$ for all $\ket v\in I(x)$ while $\langle \bar w_x| t\rangle=1$. 

Fixing a positive witness for each $x\in f^{-1}(1)$ and a negative witness for each $x\in f^{-1}(0)$ we denote their collection by $w$ and $\bar w$ respectively. Then we define the complexity of $(P, w, \bar w)$ by
\begin{align}\label{eq:wsize-b-q-0}
\wsize(P, w, \bar w) = \max \Big\{ \max_{x\in f^{-1}(1)}\big\| \ket{w_x}\big\| ^2, \max_{x\in f^{-1}(0)} \big\|  A^\dagger \ket{\bar{w}_x}\big\| ^2\Big\}.
\end{align}

In a span program we sometimes also have a set of \emph{free vectors} $I_{\text{free}}\subseteq \mathbb C^d$ that are always available. That is, we let 
$$I(x)=\Big(\bigcup_{j=1}^n I_{j,x_j}\Big)\cup I_{\text{free}}.$$
Then if $x\in f^{-1}(x)$ in writing the target vector $\ket t$ as a linear combination of available vectors, we can of course use elements of $I_{\text{free}}$ as well. Yet, since these vectors and all of whose linear combinations are freely available, we will not count their coefficients in~\eqref{eq:wsize-b-q-0}. That is, $\| \ket{w_x}\| ^2$ on the right hand side is replaced by
$$\sum_{\ket v\in\bigcup_{j=1}^n I_{j,x_j} } |\bra v w_x\rangle|^2.$$
In the case where $x\in f^{-1}(0)$, the negative witness must be orthogonal to all available vectors, and then to $I_{\text{free}}$. Thus the free vectors will not contribute to $\big\|A^\dagger \ket{\bar{w}_x}\big\| ^2$ automatically.

\begin{remark}\label{re:freevec}Including free vectors $I_{\text{free}}$ although sometimes helps in designing span programs, theoretically they do not help to improve span programs. Indeed, having a span program with free vectors, replacing the underlying vector space $\mathbb C^d$ with the quotient $\mathbb C^d/(\text{span}\, I_{\text{free}})$ and replacing each vector $\ket v\in \mathbb C^d$ with its image in the quotient ($\ket v + \text{span}\, I_{\text{free}}$), we obtain an equivalent span program with the same complexity.
\end{remark}


\begin{figure}
\centering \includegraphics[scale=1]{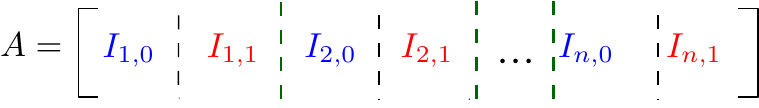}
\figcaption{Representation of the matrix $A$.}
\label{fig:Aspanprogram}
\end{figure}

Reichardt in~\cite{Rei09} showed that any feasible solution to~\eqref{eq:dual-SDP} can be transformed into a span program $(P,w,\bar{w})$ evaluating the same binary function $f:D_f \rightarrow \{0,1\}$, such that the complexity of $(P, w, \bar w)$ is equal to the objective value of~\eqref{eq:dual-SDP}. Conversely, for any span program $(P,w,\bar{w})$ there exists a feasible solution of~\eqref{eq:dual-SDP}  which has the objective value equal to the complexity of $(P, w, \bar{w})$. Therefore, designing optimal quantum query algorithms for a function $f$ with binary input and binary output ($\ell=m=2$) is equivalent to finding span programs for $f$  with the minimum complexity.

\section{Non-binary Span Program}\label{sec:NBSP}

Our main result in this paper is the generalization of the framework of span programs for non-binary functions. We show that any feasible point of the dual adversary bound corresponds to a non-binary span program and visa versa. This answers an open question first raised in~\cite{Rei09}.


A \emph{non-binary span program} (NBSP) denoted by $P$ evaluating a function $f:D_f\rightarrow[m]$ with $D_f\subseteq [\ell]^n$ consists of:
\begin{itemize}
\item a finite dimensional inner product space of the form 
$$H=H_1\oplus H_2\oplus\ldots H_n\oplus H_{\rm free},$$ 
where each $H_j$ for $1\leq j\leq n$ can be written as $H_j=H_{j,0}+\ldots+H_{j,\ell-1}$. Note that here $H_{j, q}$'s for different values of $q$ are subspaces of $H_j$ that are \emph{not necessarily orthogonal}.
\item a finite-dimensional inner product  vector space $V$,
\item $m$ non-zero vectors $|t_0 \rangle, |t_1 \rangle,\ldots ,|t_{m-1} \rangle\in V$,
\item a linear operator $A: H\to V$.
\end{itemize}
Given these data, for any $x\in D_f$, let 
$$H(x):=\bigoplus_{j\in\{1,\ldots,n\}}H_{j,x_j}\oplus H_{\rm free}\subseteq H.$$ 
Then $AH(x)$, i.e., the image of $H(x)$ under $A$, is called the \emph{space of available vectors} for $x$.
We say that $P$ evaluates $f$ if for every $x\in f^{-1}(\alpha)$ there exists 
\begin{itemize}
\item a positive witness $\ket{w_x}\in H(x)$ such that $A\ket{w_x}=\ket{t_\alpha}$ and
\item a negative witness $\ket{\bar{w}_x}\in V$ such that $\bra{\bar{w}_x}AH(x)=0$ and $\bra{\bar{w}_x}t_\beta\rangle=1-\delta_{\alpha,\beta}$ for all $\beta\in[m]$.
\end{itemize}
Here the first condition guarantees that $\ket{t_{f(x)}}$ belongs to the space of available vectors, and the second condition says that $\ket{t_\beta}$ for $\beta\neq f(x)$ does not belong to the space of available vectors. 

Now given the positive and negative witnesses for every $x\in D_f$, we define
$$\mathrm{wsize}_x(P,w_x,\bar{w}_x):=\max\big\{\|  \Pi_{\text{non-free}}\ket{w_x}\| ^2,\|  A^\dagger \ket{\bar{w}_x}  \| ^2\big\}.$$
where $\Pi_{\text{non-free}}$ is the orthogonal projection on the subspace  $H_{\text{non-free}}=\bigoplus_{j=1}^n H_j$.
Finally, the \emph{complexity} of $(P, w, \bar w)$ is defined by
\begin{align*}
 \mathrm{wsize}(P,w,\bar{w}):= \max_{x\in D_f} \; \mathrm{wsize}_x(P,w_x,\bar{w}_x). 
\end{align*}

\vspace{.2in}
We now consider a special yet important subclass of non-binary span programs. Suppose that the subspaces $H_{j, q}\subseteq H_j$ for different values of $q$ (but fixed $j$) are orthogonal to each other. In this case these subspaces are determined by orthonormal bases. Moreover, the map $A$ is determined once we fix the images of these orthonormal bases. To this end, we may fix orthonormal bases $\mathcal B_{j, q}\subset H_{j, q}$ and $\mathcal B_{\text{free}}\subset H_{\text{free}} $ and let 
\begin{align}\label{eq:NBSP-special-case}
I_{j, q} = A\mathcal B_{j, q}, \qquad I_{\text{free}} = A \mathcal B_{\text{free}},
\end{align}
as subsets of $V$. In this case we obtain the following notion of non-binary span program which we call \emph{non-binary span program with orthogonal inputs} (NBSPwOI). 


An NBSPwOI denoted by $P$ evaluating a function $f:D_f\rightarrow[m]$ with $D_f\subseteq [\ell]^n$ consists of
\begin{itemize}
\item a finite-dimensional inner product vector space $V$,
\item $m$ non-zero vectors $|t_0 \rangle, |t_1 \rangle,\ldots ,|t_{m-1} \rangle\in V$,
\item input vector sets $I_{j,q}\subseteq V$ for all $1\leq j\leq n$ and $q\in [\ell]$,
\item and a set of free vectors $I_{\text{free}}\subseteq V$.
\end{itemize}
We then define $I\subseteq V$ by 
$$I=I_{\text{free}}\cup \bigcup_{j=1}^n \bigcup_{q\in [\ell]} I_{j,q}.$$
Moreover, for each $x\in D_f$ we define the set of \emph{available vectors} by 
$$I(x)=I_{\text{free}}\cup\bigcup_{j=1}^nI_{j,x_j}.$$
That is, vectors in $I_{j, q}$ become available when the $j$-th coordinate of $x$ is $q$.
Also, as in the binary case, 
we let $A$ be the $d\times |I|$ matrix consisting of all input vectors as its columns.

We say that a non-binary span program with orthogonal inputs $P$ evaluates the function $f$ if for each $x\in D_f$ with $f(x)=\alpha\in [m]$, $\ket{t_\alpha}$ belongs to the span of the available vectors $I(x)$ and $\ket{t_{\beta}}$, for $\beta\neq \alpha$ does not belong to the span of $I(x)$. Even more, there should be two witnesses indicating these. Namely a positive witness  $\ket{w_x} \in \mathbb C^{|I|}$ and a negative witness $\ket{\bar{w}_x}\in \mathbb C^d$ satisfying the following: 
\begin{itemize}
\item First, the coordinates of $\ket{w_x}$ associated to unavailable vectors are zero. 
\item Second, $A\ket{w_x} = \ket{t_\alpha}$. 
\item Third, for all $\ket v\in I(x)$ we have $\langle v|\bar w_x\rangle =0$. 
\item Fourth, for all $\beta\neq \alpha$ we have $\langle t_{\beta}|\bar{w}_x\rangle=1$. 
\end{itemize}

The complexity of the non-binary span program with orthogonal inputs $P$ together with the collection $w$ and $\bar w$ of positive and negative witnesses is defined similar to that of NBSP. For every $x\in D_f$ we define
$$\mathrm{wsize}_x(P,w_x,\bar{w}_x):=\max\big\{\|  \Pi_{\text{non-free}}\ket{w_x}\| ^2,\|  A^\dagger \ket{\bar{w}_x}\| ^2\big\},$$
where $\Pi_{\text{non-free}}$ is the projection on the coordinates associated to non-free vectors.
Next the complexity of $(P, w, \bar w)$ equals
\begin{align*}
 \mathrm{wsize}(P,w,\bar{w}):= \max_{x\in D_f} \; \mathrm{wsize}_x(P,w_x,\bar{w}_x). 
\end{align*}

\vspace{.2in}
As mentioned before, NBSPwOI is a special case of NBSP when the subspaces $H_{j, q}$ for different values of $q$, are orthogonal to each other, and the correspondence between these two definitions is given by~\eqref{eq:NBSP-special-case} where $\mathcal B_{j, q}\subset H_{j, q}$ and $\mathcal B_{\text{free}}\subset H_{\text{free}} $ are orthonormal bases.


Some remarks are in line regarding the  definition of non-binary span program:

\begin{remark}
Observe that in the non-binary span program we need a positive witness and a negative witness for every $x\in D_f$. This is unlike the binary case in which we need a positive witness only for $x\in f^{-1}(1)$ and a negative witness only for $x\in f^{-1}(0)$. Thus it is not immediate from the definitions that our non-binary span program is a generalization of the binary one. We will show this fact in Lemma~\ref{lem:binary-nonbinary-equiv} below. 
\end{remark}

\begin{remark}
Notice that $\ket{\bar{w}_x}$ is independent of $\beta$, i.e., we have a single negative witness $\ket{\bar{w}_x}$ indicating that $\ket{t_{\beta}}$ does not belong to the span of available vectors for all $\beta\neq \alpha=f(x)$. This is stronger than saying that $\ket{t_{\beta}}$ for $\beta\neq \alpha$ is not in the span of available vectors.
\end{remark}

\begin{remark}\label{rem:all-beta}
The fact that the negative witness is independent of $\beta\neq \alpha=f(x)$ mentioned above, may seem to be a very strong condition. Nevertheless, we argue that this is the right definition. 
It is shown in~\cite{Rei09} that any quantum query algorithm for a binary function with \emph{one-sided error} can be converted to a span program. If we adopt the same techniques and generalize it for functions with non-binary input/output, the resulting span program satisfies this condition and matches our definition. For details see Appendix~\ref{app:qa2non-sp}. Later we will give another argument for this choice of definition.
\end{remark}


\begin{lemma}\label{lem:binary-nonbinary-equiv}
The non-binary span program with orthogonal inputs, when restricted to binary functions, is a generalization of Reichardt's span program described in Subsection~\ref{sec:SP}.
\end{lemma}

\begin{proof}
We need to show that for any binary span program $P$ for a function $f:D_f\to \{0,1\}$, with $D_f\subseteq \{0,1\}^n$ as defined in Subsection~\ref{sec:SP}, there exists a corresponding non-binary span program with the same complexity. 

Let $P$ be a binary span program for $f$. Reichardt in Lemma 4.1 of~\cite{Rei09} has shown that associated to $P$ there exists a binary span program $P'$ for the function $\bar f=1-f$ with the \emph{same complexity}. This fact can also be concluded using the symmetry in the definition of the quantum query complexities of $f$ and $\bar f$ and the fact that they are characterized by binary span programs, yet there is an explicit construction in~\cite{Rei09} for such a span program $P'$. Having these two binary span programs $P$ and $P'$, we construct a \emph{non-binary} span program for $f$.

Let the binary span program $P$ be determined by the vector space $V$, subsets $I_{j, q}\subseteq V$, for $1\leq j\leq n$ and $q\in \{0,1\}$, target vector $\ket t$ and positive and negative witnesses $\ket {w_x}$ and $\ket{\bar w_y}$ for $x\in f^{-1}(1)$ and $y\in f^{-1}(0)$ respectively. Also, denote these parameters for $P'$ by $V'$, $I'_{j, q}$, etc. 
Let the vector space of our non-binary span program be $V''= V\oplus V'$ and consider the natural embedding of $V, V'$ in $V''$. Define $I''_{j, q} = I_{j, q}\cup I'_{j, q}$ where by abuse of notation we consider $I_{j, q}$ and $I'_{j, q}$ as subsets of $V''$ as well. Also we let $\ket{t_1} = \ket t$ and $\ket{t_0}=\ket{t'}$ be the target vectors of the non-binary span program. Finally for any $x\in f^{-1}(1)$ we let
$$\ket{w''_x} = \ket{w_x}, \qquad \ket{\bar{w}''_x} = \ket{\bar w'_x},$$
and for any $x\in f^{-1}(0)$ we let 
$$\ket{w''_x} = \ket{w'_x}, \qquad \ket{\bar{w}''_x} = \ket{\bar w_x}.$$
Then it is easy to verify that these define a valid non-binary span program with orthogonal inputs evaluating $f$ with the same complexity as that of $P$.
\end{proof}


The next step is to show that NBSP characterizes the quantum query complexity of non-binary functions. 
 In order to do so, as in the case of span program for binary functions~\cite{Rei09}, we need to define \emph{canonical} non-binary span programs.

\begin{definition}[Canonical NBSP] 
A canonical non-binary span program $(P,w,\bar{w})$ evaluating $f: D_f\rightarrow[m]$ with $D_f\subseteq [\ell]^n$ is a non-binary span program satisfying the followings:
\begin{itemize}
\item  We have $\dim V=|D_f|$ and an orthonormal basis for $V$ can be indexed by elements of $D_f$. We denote this orthonormal basis by $\big\{\ket{e_x}:\, x\in D_f\big\}$. 
\item The target vector $\ket{t_\alpha}$ for every $\alpha\in [m]$ is given by 
$$|t_\alpha\rangle=\sum_{y: f(y)\neq \alpha} |e_y\rangle.$$ 
\item For all $x\in D_f$ the negative witness $\ket{\bar{w}_x}$ equals $\ket{e_x}$.
\item $ H_{\text{\rm free}} =\{0\}$
\end{itemize}
\end{definition}
Now we show that every non-binary span program has an equivalent canonical span program.

\begin{proposition}\label{pro:SpanProgram2can}
For any non-binary span program $(P,w,\bar{w})$ evaluating $f: D_f\rightarrow[m]$ with $D_f\subseteq [\ell]^n$ there exists a \emph{canonical} non-binary span program $(P',w',\bar{w}')$ evaluating the same function with the same complexity.
\end{proposition}
\begin{proof}
Consider the linear transformation $B:V\rightarrow\mathbb{C}^{|D_f|}$ given by 
$$B=\sum_{y\in D_f} |e_y\rangle \langle \bar{w}_y|.$$ 
We obtain $P'$ by transforming the operator $A$ and vectors $\ket{t_\alpha}$'s using $B$. That is,  we let $H'_{j, q} = H_{j, q}$ and $H'=\bigoplus_j H'_j$.  Then, $A': H'\to C^{|D_f|}$ is given by $A'=BA$. Note that by definition $\ket{\bar w_x}$, for all $x$, is orthogonal to $AH_{\text{free}}$. Thus $BAH_{\text{free}} = \{0\}$, and with no loss of generality we can assume that $H'_{\text{free}}=\{0\}$. Next, the target vectors $|t'_\alpha\rangle$ of $P'$ are
\begin{equation*}
\ket{t'_\alpha}=B\ket{t_\alpha}=\sum_{y\in D_f} \ket{e_y} \langle \bar{w}_y|t_\alpha\rangle=\sum_{y:f(y)\neq \alpha} |e_y\rangle,
\end{equation*}
as desired. Now let $\ket{w'_x} = \Pi_{\text{non-free}} \ket{w_x}$. We have $\ket{w'_x}\in H'(x)$ and because of $A\ket{w_x}=|t_\alpha\rangle$ (and that $BA H_{\text{free}} =\{0\}$) we have $A'\ket{w'_x}=B|t_\alpha\rangle = \ket{t'_\alpha}$. Therefore,  the positive witnesses of $P'$ remain the same since $\|\ket{w'_x}\|=\| \Pi_{\text{non-free}}|w_x\rangle\|$. 
Now we need to verify that $\ket{\bar w'_x}=\ket{e_x}$ are valid negative witnesses for $P'$. We have
$$\bra{\bar{w}'_x}BAH'(x)=\bra{e_x}\sum_{y\in D_f}\ket{e_y}\bra{\bar{w}_y}AH'(x)=\bra{\bar{w}_x}AH'(x)=0,$$
and also, from the definitions of $\ket{t'_\alpha}$ for every $\beta\neq \alpha=f(x)$ we have  $\langle e_x|t'_\beta\rangle=1$.
We conclude that $(P', w', \bar w')$ is a canonical non-binary span program evaluating $f$. Moreover, the complexity of $(P', w',\bar w')$ is the same as that of $(P, w, w')$.  This is because the positive witness sizes remain the same in $P'$ and 
we have
\begin{equation*}
\big\|\bra{\bar w'_x} A'\big\| ^2=\big\| \bra{e_x} BA \big\| ^2=\Big\| \bra{e_x} \sum_{y\in D_f}  \ket{e_y}\bra{\bar{w}_x}A \Big\| ^2=\big\|\bra{\bar{w_x}}  A\big\| ^2 .\qedhere
\end{equation*}
\end{proof}

\begin{remark}
In the above proof starting with a non-binary span program with orthogonal inputs, we obtain a canonical non-binary span program with orthogonal inputs. This is simply because the subspaces $H_{j, q}$ remain the same in the canonical span program and remain orthogonal for different values of $q$ if they are orthogonal in the starting NBSP. 
\end{remark}

Now we are ready to prove the main result of this section, that (canonical) non-binary span programs are equivalent to solutions of the dual adversary SDP. 

\begin{theorem}\label{thm:NBSP2ADV}
Non-binary span programs characterize the quantum query complexity of functions with non-binary input and/or output, up to a constant factor. 
\end{theorem}

\begin{proof}
To prove this theorem we show that for any non-binary span program $(P,w,\bar{w})$ evaluating a function $f:D_f\rightarrow [m]$ with $D_f\subseteq [\ell]^n$, there exists a feasible solution to the dual adversary SDP~\eqref{eq:dual-SDP} with objective value being equal to the complexity of $(P, w, \bar{w})$ and vise versa.
Using Proposition~\ref{pro:SpanProgram2can} we can assume that the span program $(P,w,\bar{w})$ is canonical. Therefore $A:\bigoplus_j H_j\to V$ with $\dim V=|D_f|$ can be written as 
$$A=\sum_{j,x} \ket{e_x}\bra{a_{x,j}},$$
with $\ket{a_{x, j}}\in H_j$. Moreover,  for every $x\in D_f$ we have $\ket{w_x}=\oplus_j\ket{w_{x, j}}$ with $\ket{w_{x, j}}\in H_{j, x_j}$ and $A\ket{w_x} = \sum_{y: f(y)\neq f(x)} \ket{e_y}$. Now define 
\begin{equation}\label{eq:feasible-sdp-00}
\ket{u_{x,j}}=\ket{a_{x,j}}, \qquad \ket{v_{x,j}}=\ket{w_{x,j}} \qquad \forall x\in D_f, j\in \{1,\ldots,n\}.
\end{equation}
Let $x,y\in D_f$ and $j_0\in \{1,\dots, n\}$ be such that $x_{j_0}=y_{j_0}$. Then $\ket{w_{x,{j_0}}}\in H_{{j_0}, x_{j_0}}\subseteq H(y)$ and $\ket{\bar w_y} = \ket{e_y}$ is orthogonal $AH(y)$. This means that 
$\bra{e_y} A\ket{w_{x,{j_0}}}=0$ and equivalently $\bra {a_{y,{j_0}}} w_{x, {j_0}}\rangle=0$. Therefore,
\begin{equation*}
\sum_{j:x_j\neq y_j}\bra{u_{y,j}}v_{x,j}\rangle=\sum_{j: x_j\neq y_j} \bra{a_{y,j}} w_{x,j}\rangle  =\sum_{j=1}^n \bra{a_{y,j}} w_{x,j}\rangle = \bra{e_y} A\ket{w_x}=1-\delta_{f(x),f(y)},
\end{equation*}
where we used the fact that $\bra{a_{y,j}}w_{x,j}\rangle=0$ for all $j$ with $x_j=y_j$. Thus~\eqref{eq:feasible-sdp-00} forms a feasible solution of the SDP~\eqref{eq:dual-SDP}.  Moreover,  it is easy to verify that the objective value of the SDP for this feasible solution coincides with the complexity of the non-binary span program.

For the other direction, we make use of two sets of vectors
$\{\ket{\mu_q}:\, q\in[\ell] \}$ and $\{\ket{\nu_q}:\, q\in [ \ell]\}$ in $\mathbb C^{\ell}$ first appeared in~\cite{LMRS10}:
\begin{align}\label{eq:mu}
&\ket{\mu_q}=\sqrt{\frac{2(\ell-1)}{\ell}}\left(-\theta\ket{q}+\frac{\sqrt{1-\theta^2}}{\sqrt{\ell-1}}\sum_{p\neq q}\ket{p}\right), \\ \label{eq:nu}
&\ket{\nu_q}=\sqrt{\frac{2(\ell-1)}{\ell}}\left(\sqrt{1-\theta^2}\ket{q}+\frac{\theta}{\sqrt{\ell-1}}\sum_{p\neq q}\ket{p}\right),
\end{align}
where $\theta=\sqrt{\frac12-\frac{\sqrt{\ell-1}}{\ell}}$. These vectors have the property that $\|\ket{\mu_q}\|^2=\|\ket{\nu_q}\|^2=\frac{2(\ell-1)}{\ell}\leq 2$ for all $q$ and we have
$$\bra {\mu_q} \nu_p \rangle=1-\delta_{q, p}.$$

Let the vectors $\ket{u_{x,j}},\ket{v_{x,j}}\in U$ form a feasible solution of the dual adversary SDP~\eqref{eq:dual-SDP}. Define a non-binary span program as follows: 
\begin{itemize}
\item   $H=U\otimes \mathbb{C}^n\otimes \mathbb{C}^\ell$ and $H_j = U\otimes \Span \{\ket j\}\otimes \mathbb C^\ell$ for $j=1, \dots, n$,
\item  $H_{j,q}=U\otimes \Span\{\ket{j}\}\otimes \ket{\nu_q}^\perp \subseteq H_j$  for every $1\leq j\leq n$ and $q\in[\ell]$, where $\ket{\nu_q}^\perp\subset \mathbb C^\ell$ is the subspace of vectors  orthogonal to $\ket{\nu_q}$,
\item  $V={\rm span}\{\ket{e_x}:x\in D_f\}$,
\item $\ket{t_\alpha }=\sum_{x:f(x)\neq \alpha} \ket{e_x}$,
\item and $A: H\to V$ is given by 
$$A\ket{u,j,q}=\sum_{y\in D_f}\bra{ v_{y,j}}u\rangle\bra{\nu_{y_j}}q\rangle\ket{e_y}.$$
\end{itemize}
Now for every $x\in D_f$, let $\ket{w_x}=\sum_j\ket{u_{x,j},j,\mu_{x_j}}$. We have
\begin{equation*}
A\ket{w_x}=\sum_{y\in D_f}  \sum_{j=1}^n\bra{v_{y,j}}u_{x,j}\rangle\bra{\nu_{y_j}}\mu_{x_j}\rangle\ket{e_y}=\sum_{y\in D_f}  \Big(\sum_{j:x_j\neq y_j}\bra{v_{y,j}}u_{x,j}\rangle\Big) \ket{e_y}=\ket{t_{f(x)}},
\end{equation*}
where for the second equality we use $\bra{\nu_{q}}\mu_{q'}\rangle = 1-\delta_{q, q'}$.
Next, let $\ket{\bar{w}_x}=\ket{e_x}$, then $\bra{\bar{w}_x}t_\alpha\rangle=1-\delta_{f(x),\alpha}$ and for every $z\in\ket{\nu_{x_j}}^\perp$ we have
\begin{equation*}
\bra{\bar{w}_x}A\ket{u,j,z}=\bra{v_{x,j}}u\rangle\left.\bra{\nu_{x_j}}z\right\rangle=0.
\end{equation*}
Therefor, $\bra{e_x} AH(x)=0$ and this non-binary span program evaluates $f$.
Finally, the complexity of this NBSP is given by 
\begin{align*}
&\max_{x\in D_f}\; \max\big\{\big\| A^\dagger \ket{\bar{w}_x} \big\| ^2, \big\| |w_x \rangle \big\| ^2\big\},\\
=&\max_{x\in D_f}\max\left\{\sum_{j=1}^n \left\|\ket{v_{x,j}}\ket{\nu_{x_j}}\right\|^2, \sum_{j=1}^n\left\|\ket{u_{x,j}}\ket{\mu_{x_j}}\right\|^2  \right\}\\
\leq& 2 \max_{x\in D_f}\max\left\{\sum_{j=1}^n\|\ket{u_{xj}}\|^2,\sum_{j=1}^n\|\ket{v_{xj}}\|^2\right\}
\end{align*}
which is equal to the objective value of the dual adversary SDP~\eqref{eq:dual-SDP} up to a constant factor.
\end{proof}

In the above proof we transform any NBSP, and in particular any NBSPwOI to a feasible solution of the dual adversary SDP. However, transforming a feasible solution of the dual adversary SDP to an NBSP, the resulting non-binary span program is \emph{not} an NBSPwOI. In the following proposition we see that we can indeed transform a feasible solution of the dual adversary SDP to an NBSPwOI, but by losing a factor of $\sqrt{\ell-1}$.


\begin{proposition} \label{pro:SDP2SPwob}
Any feasible solution of the dual adversary SDP~\eqref{eq:dual-SDP} can be transformed to a canonical non-binary span program \emph{with orthogonal inputs} $(P,w,\bar{w})$ evaluating the same function $f:D_f\rightarrow [m]$ such that 
the complexity of $(P, w, \bar{w})$ is equal to $\sqrt{\ell-1}$ times the objective value of the dual adversary SDP. 
\end{proposition}

See Appendix~\ref{app:SDP2SPwob} for a proof of this proposition.

\begin{remark}\label{rem:NBSP-vs-wOI}
The definition of NBSPwOI is more intuitive comparing to general NBSPs and consistent with its binary counterpart, yet as the above proposition suggests and will be shown later in an example, NBSPwOI is optimal only up to a factor of $\sqrt{\ell-1}$. However, when $\ell$ is a constant, it is more reasonable to think of NBSPwOI rather than general NBSPs. Moreover, sometimes getting ideas from a non-optimal NBSPwOI, we can design an NBSP that is optimal. Later, we will see examples of such non-binary span programs.

\end{remark}

\begin{remark} \label{rem:gamma-SP-complexity}
For any non-binary span program $(P, w, \bar w)$ by scaling the associated map $A$ and positive witnesses $\ket{w_x}$ by factors $\gamma^{-1}$ and $\gamma$ respectively (but leaving the target vectors and negative witnesses unchanged), we get another non-binary span program $(P', w', \bar w')$ for the same function. The complexity of this new span program equals 
\begin{align*}
\wsize(P', w', \bar w') & = \max_x \max\big\{   \gamma^2 \big\|  \ket{w_x}\big \| ^2, \gamma^{-2} \big\|  A^\dagger\ket{\bar{w}_x}\big\| ^2\big\}\\
& = \max \Big\{ \gamma^2 \max_x \big\|  \ket{w_x}\big \| ^2,\, \gamma^{-2}\max_x    \big\|  A^\dagger\ket{\bar{w}_x}\big\| ^2   \Big\}
\end{align*}
Then letting 
$$\gamma^2 = \frac{ \max_x\big\|  A^\dagger\ket{\bar{w}_x}\big\| }{\max_x\big\|  \ket{w_x}\big \|},$$
we obtain 
$$\wsize(P', w', \bar w') = \big(\max_x\big\|  \ket{w_x}\big \|\big)\big(\max_x\big\|  A^\dagger\ket{\bar{w}_x}\big\|\big).$$
Thus we may define the \emph{positive} and \emph{negative witness sizes} as 
$$\wsize_+(P, w) : = \max_x\big\|  \ket{w_x}\big \|^2, \qquad \quad \wsize_-(P, \bar w) : = \max_x\big\|  A^\dagger\ket{\bar{w}_x}\big\|^2,$$
and let the complexity of the span program to be 
$$\sqrt{\wsize_+(P, w) \cdot \wsize_-(P, \bar w)}\,.$$
See the proof of Proposition~\ref{pro:SDP2SPwob} for an example.
\end{remark}

\begin{remark} 
One may suggest that if we convert a non-binary span program directly to a quantum algorithm (and not to a solution of the dual adversary bound as we did in Proposition~\ref{pro:SDP2SPwob}) we can relax the condition that for all $\beta\neq f(x)$ we have $\bra{\bar{w_x}}t_\beta\rangle=1$ (see also Remark~\ref{rem:all-beta}).
This is based on the same techniques that has been used by Reichardt~\cite{Rei09} in the binary case\footnote{ The algorithm is just doing phase estimation on the product of two reflections, one of them around the kernel of $[\ket{t_0}\ldots \ket{t_{m-1}}|A]$ and the other one around $\Pi_x=1-\sum_{j\in [n]}\sum_{a \neq x_j}\ket{j,a}\bra{j,a}$.} . Nevertheless, in this proof we again see that this strong condition must be satisfied in order to proof go through. 
\end{remark}


\begin{example}\label{ex:id}
Let $D$ be the subset of $[\ell]^n$ containing those $x\in [\ell]^n$ that have at most one non-zero coordinate. Let $f: D\to D$ be the identity function on $D$. Then the optimal non-binary span program with orthogonal inputs evaluating this function has complexity $\Theta\big(\sqrt{(\ell-1)n}\big)$.
\end{example}


Note that the function in this example is a non-binary generalization of the OR function, and using the Grover search algorithm its quantum query complexity is $\Theta(\sqrt n)$. Thus, this example shows that the undesirable $\sqrt{\ell-1}$ factor in the Proposition~\ref{pro:SDP2SPwob} is necessary and cannot be improved. 


\begin{proof}
We first present a non-binary span program with orthogonal inputs for $f$ with complexity $O(\sqrt{(\ell-1)n})$ and then prove its optimality. 

For simplicity of notation let us index elements of $D$ by $\{(0)\}\cup\{(j, q):\, 1\leq j\leq n, 0\neq q\in [\ell]\}$. Thus $x_{(0)} = (0, \dots, 0)$ and $x_{(j, q)}$ is the sole element of $D$ whose $j$-th coordinate is $q\neq 0$. Here is the span program with orthogonal inputs:
\begin{itemize}
\item An orthonormal basis for the vector space of our span program is $\{\ket{j, q}:\, 1\leq j\leq n, q\in [\ell]\}$.
\item   $I_{j, q} = \{\ket{j, q}\}$. 
\item The target vectors are 
 $$\ket{t_{(0)}}= a  \sum_{i=1}^n \ket{i, 0},$$
 and
 $$\ket{t_{(j, q)}} = b\sum_{i:i\neq j} \ket{i, 0}  + c\ket{j, q}, \qquad \forall j, q\neq 0,$$
 where $a, b, c>0$ are such that $a^2= \sqrt{(\ell-1)/n}$, $b^2 = \sqrt{(\ell-1)n}/(n-1)$ and $c^2=\sqrt{(\ell-1)n}$.
 \end{itemize}
Clearly, the positive witness sizes for $x_{(0)}$ and $x_{(j, q)}$ are $na^2 = \sqrt{(\ell-1)n}$ and 
$(n-1)b^2+c^2 = 2\sqrt{(\ell-1)n}$, respectively, both of which are $O(\sqrt{(\ell-1)n})$.
The negative witnesses are
$$\ket{\bar w_{(0)}} = \frac{1}{c}\sum_i \sum_{p\neq 0} \ket{i, p},$$
and
$$\ket{\bar w_{(j, q)}} = \frac{1}{a} \ket{j, 0} + \frac{1}{c} \sum_{p\notin\{0, q\}} \ket{j, p} + \frac{1-b/a}{c} \sum_{i\neq j} \sum_{p\neq 0} \ket{i, p}.$$
Clearly, $\ket{\bar w_{(0)}}$ is orthogonal to all available vectors when the input is $x= x_{(0)}$ and $\bra{\bar w_{(0)}} t_{(j, q)}\rangle =1$ for all $j$ and $q\neq 0$. Moreover, 
$$\sum_{i, p} |\langle i, p \ket{\bar w_{(0)}} |^2=(\ell-1)n/c^2 =\sqrt{(\ell-1)n}.$$
Similarly, $\ket{\bar w_{(j, q)}}$ is orthogonal to all available vectors when the input is $x_{(j, q)}$. Also $\bra{\bar w_{(j, q)}} t_{(0)}\rangle=1$, and for $p\neq q$ we have  $\bra{\bar w_{(j, q)}}  t_{(j, p)}\rangle  =1$. Moreover, for $i\neq j$ and arbitrary $p$ we have
$$\bra{\bar w_{(j, q)}} t_{(i, p)}\rangle =\frac{b}{a} + \big(1-\frac{b}{a}\big) =1.$$
Finally we have
$$\sum_{i, p} |\langle i, p \ket{\bar w_{(j, q)}} |^2= \frac{1}{a^2} + (\ell-2) \frac{1}{c^2} + (n-1)(\ell-1) \Big(  \frac{1-b/a}{c} \Big)^2 \leq 5 \sqrt{(\ell-1)n}.$$
Thus the complexity of this span program is $O(\sqrt{(\ell-1)n})$.

We now prove the optimality of this bound.  By Proposition~\ref{pro:SpanProgram2can} it suffices to consider only canonical span programs with orthogonal inputs which are determined by the set of input vectors and positive witnesses.  Let $I_{j, q} =\{\ket{v_{j,q,1}}, \ldots, \ket{v_{j,q,k_{j,q}}}\}$. Also, assume that for every $x\in D$

\begin{equation}\label{eq:tz}
\ket{t_{x}}=\sum_{y\neq x}\ket{e_y}=\sum_{i=1}^n \sum_{r=1}^{k_{i,x_i}} w_{x, i ,r}\ket{v_{i,x_i,r}},
\end{equation}
letting $\ket{w_x}$ be the vector of coefficients in the above sum, the positive witness size is equal to 
\begin{equation}\label{eq:Pz}
S_x:=\big\|\ket{w_x}\big\|^2=\sum_{i=1}^n\sum_{r=1}^{k_{i,x_i}} \big|w_{x, i,r}\big|^2.
\end{equation}
Since the negative witness equals $\ket{\bar w_x}=\ket{e_x}$, we have
\begin{equation}\label{eq:N-witness-orth}
\bra{e_x}v_{i,x_i,r}\rangle=0, \qquad \forall i, r,
\end{equation}
and the negative witness size for input $x$ is
\begin{equation}\label{eq:Nz}
\bar{S}_x:=\sum_{i,q}\sum_{r=1}^{k_{i,q}} \big|\bra{e_x}v_{i,q,r}\rangle\big|^2
\end{equation}
From~\eqref{eq:tz} for every $x\neq y\in D$ we have 
\begin{equation*}
1=\sum_{i=1}^n \sum_{r=1}^{k_{i,x_i}} w_{x, i ,r}\langle e_y\ket{v_{i,x_i,r}}=\langle w_x\ket{ \alpha_{y, x}},
\end{equation*}
where $\ket{\alpha_{y,x}}$ is the vector of coefficients $\langle e_y\ket{v_{i,x_i,r}}$.
Therefore, by the Cauchy-Schwarz inequality
\begin{equation}\label{eq:w-alpha-x-ineq}
\big\|\ket{w_x}\big\| \cdot \big\| \ket{\alpha_{y,x}}\big\|\geq \big|\bra{w_x}\alpha_{y,x}\rangle\big|=1.
\end{equation}
On the other hand for $x=x_{(0)}\in D$ using~\eqref{eq:N-witness-orth} we have
\begin{align}
\sum_{y: y\neq x_{(0)}} \big\| |\alpha_{x_{(0)},y}\rangle\big\|^2 & =\sum_{y: y\neq x_{(0)}}\sum_i\sum_{r=1}^{k_{i,y_i}} \Big|\big\langle e_{x_{(0)}}\ket{v_{i,y_i,r}}\Big|^2\nonumber\\
&=\sum_{i,q\neq 0}\sum_r \big| \langle e_x\ket{v_{i,q,r}} \big|^2\nonumber\\
&=\bar{S}_{x_{(0)}}.\label{eq:Sprime_x}
\end{align}
As a result, letting $S= \max_x S_x$ and $\bar{S}=\max_x \bar{S}_x$, the complexity of the non-binary span program is lower bounded by
\begin{align*}
\max\{S,\bar{S}\} & \geq \max\bigg\{\frac{\sum_{y: y\neq x_{(0)}}S_{y}}{(\ell-1)n}, \bar{S}_x \bigg\}\\
&\geq\sqrt{\left(\frac{\sum_{y: y\neq x_{(0)}}S_{y}}{(\ell-1)n}\right) \bar{S}_{x_{(0)}}}\\
&=\frac{1}{\sqrt{(\ell -1)n}}\sqrt{\sum_{y: y\neq x_{(0)}}\big\|\ket{w_y}\big\|^2} \sqrt{\sum_{y: y\neq x_{(0)}} \big\| |\alpha_{x_{(0)},y}   \rangle \big\|^2}\\
&\geq \frac{1}{\sqrt{(\ell -1)n}}\sum_{y:y\neq x_{(0)}}   \big\|\ket{w_y}\big\| \cdot \big\| |\alpha_{x_{(0)},y}\rangle\big\|  \\
&\geq \frac{(\ell-1)n}{\sqrt{(\ell -1)n}}\\
&=\sqrt{(\ell-1)n},
\end{align*}
where in the third line we use~\eqref{eq:Sprime_x}, in the fourth line we use  Cauchy-Schwarz inequality and in the fifth line we use~\eqref{eq:w-alpha-x-ineq}.
Therefore any span program with orthogonal inputs for this problem has complexity at least $\sqrt{(\ell-1)n}$.
\end{proof}

Our next example is the Max function whose quantum query complexity was first shown in~\cite{Durr1996} to be $\Theta(\sqrt n)$. Here, using non-binary span program we prove the (loose) upper bound of $O\left(\sqrt{(\ell-1)n}\right)$. Note that the lower bound of $\Omega(\sqrt n)$ is immediate since Max is a generalization of the OR function.

\begin{example}[Max function]\label{ex:max}  The quantum query complexity of the function $\mathrm{Max}_n:[\ell]^n\to[\ell]$,  that given a list of $n$ numbers in $[\ell]$ outputs the maximum element of the list, is $O\left(\sqrt{(\ell-1)n}\right)$.
\end{example}


\begin{proof}
We first construct an NBSPwOI for this function with complexity $O\left((\ell-1)\sqrt{n}\right)$ and then convert it to an NBSP with complexity $O\left(\sqrt{(\ell-1)n}\right)$.
The non-binary span program with orthogonal inputs is as follows:
\begin{itemize}
\item The vector space is $(n+1)\ell$-dimensional with orthonormal basis
$$\{\ket{j,q}:\, q\in [\ell], 1\leq j\leq n\}\cup \{\ket{q}:\, q\in \ell\}.$$
\item $I_{j, q}=\{\ket{q}\}\cup\{\ket{j,r}: \;\forall r\geq q\}$
\item the target vectors are 
\begin{equation*}
\ket{t_\alpha}=c\ket{\alpha}+c^{-1}\sum_{j=1}^n\ket{j,\alpha}
\qquad \forall \alpha\in[\ell],
\end{equation*}
where $c^2=\sqrt{\frac{n}{\ell-1}}$.
\end{itemize}
If $x\in f^{-1}(\alpha)$, there exists $j$ such that $ x_j=\alpha$. Then the vectors $\ket{\alpha}$ and $\ket{\alpha, i}$, for all $1\leq i\leq n$ are available since $x_i\leq \alpha$. Therefore, $\ket{t_\alpha}$ can be written as a linear combination of available vectors, and the positive witness size is $c^2+nc^{-2}=2\sqrt{n(\ell-1)}$. 

Let
\begin{equation*}
\ket{\bar w_x}=c^{-1}\sum_{q>\alpha}\ket{q}+c\sum_{q<\alpha} \ket{j,q}.
\end{equation*}
It is easy to verify that $\ket{\bar w_x}$ is orthogonal to all available input vectors and $\langle \bar{w_x}\ket{t_\beta}=1$ for all $\beta\neq \alpha$. Thus $\ket{\bar w_x}$ is a valid negative witness for $x$ whose negative witness size is equal to 
\begin{equation*}
n(\ell-\alpha-1)c^{-2}+\alpha^2 c^2 \leq 2(\ell-1)^\frac{3}{2}\sqrt{n}.
\end{equation*}
Therefore, based on Remark~\ref{rem:gamma-SP-complexity} the complexity of this span program equals 
$$\sqrt{2\sqrt{n(\ell-1)}\cdot 2(\ell-1)^\frac{3}{2}\sqrt{n}} = 2(\ell-1)\sqrt {n},$$
Therefore the complexity of this span program with orthogonal inputs is equal to $O\left((\ell-1)n\right).$

In the above NBSOwOI we see that the vectors in $I_{j, q}$, for different values of $q$ (but fixed $j$) are not orthogonal; indeed these sets have non-empty intersections. Thus it is natural to try to convert this  NBSPwOI to a non-binary span program. As we will see this would save a factor of $\sqrt{\ell-1}$ in the complexity. Let
\begin{itemize}
\item $H=\bigoplus_j H_j$ and $H_j={\rm span}\big\{\ket{j} \otimes \ket{q}:\, q\in[\ell]\big\}\cup \big\{\ket{j}\otimes \ket{j,q}:\, q\in[\ell]\big\}$ so that $\dim H_j = 2\ell$.
\item $H_{j,q}={\rm span}\{\ket{j}\otimes\ket{q}\}\cup\{\ket{j}\otimes\ket{j,r}: r\geq q\}$
\item $V={\rm span} \big\{\ket{q}:\, q\in [\ell]\big\} \cup\big\{\ket{j,q}:\, q\in [\ell], j\in \{1,\ldots,n\}\big\}$
\item target vectors $\ket{t_\alpha}=c\ket{\alpha}+c^{-1}\sum_{j=1}^n\ket{j,\alpha}$ are as before
\item $A:H\to V$ is given by $A=\sum_{j,q}\ket{q}\bra{j}\otimes\bra{q}+\sum_{j,q}\ket{j,q}\bra{j}\otimes\bra{j,q}$
\end{itemize}
For any $x\in {\rm Max_n}^{-1}(\alpha)$ there exists $j_0\in\{1,\ldots,n\}$ such that $x_{j_0}=\alpha$. Then the positive witness $\ket{w_x}\in H(x)$ is given by
$$\ket{w_x}=c\ket{j_0}\otimes\ket{\alpha}+c^{-1}\sum_j \ket{j}\otimes\ket{j,\alpha},$$
and the positive witness size is $c^2+c^{-2}n$. The negative witness $\ket{\bar{w}_x}\in V$ remains as before
$$\ket{\bar{w}_x}=c^{-1}\sum_{q>\alpha} \ket{q}+c\sum_{q<\alpha}\ket{j_0,q},$$
yet the negative witness would be $(\ell-\alpha-1)c^{-2}n+\alpha c^2.$
Then by setting $c^2=\sqrt{n}$ and using Theorem~\ref{thm:NBSP2ADV} and Remark~\ref{rem:gamma-SP-complexity} we have 
\begin{equation*}
Q (\mathrm{Max}_n) =\Theta\big(\mathrm{ADV}^\pm (\mathrm{Max}_n)\big)=O\left(\sqrt{(\ell-1)n}\right).\qedhere
\end{equation*}
\end{proof}

A span program based quantum query algorithm for the triangle finding problem is proposed in~\cite{BR12}. This algorithm decides whether the input graph contains a triangle or is a forest with $O(n)$ quantum queries.
In the following example, we introduce a non-binary span program for the problem of not just distinguishing the existence of a triangle, but \emph{finding} it. In this regard, to get a function (not a relation) we need to assume that the input graph contains a unique triangle.

\begin{example}[Triangle finding] There exists a quantum query algorithm that given a simple graph $G=(V,E)$ with $n$ vertices  containing no cycle except a unique triangle, outputs the vertices of the triangle using $O(n)$ queries to the edges of $G$.
\end{example}

In this example, since the input alphabet of the function is constant we may restrict ourself to NBSPwOI (see Remark~\ref{rem:NBSP-vs-wOI}).

\begin{proof}
We design a non-binary span program using ideas in~\cite{BR12}. We first randomly color vertices of $G=(V,E)$ with three colors $c: V\to\{0,1,2\}$. Under this coloring, with probability $\frac6{27}$ the vertices of the unique triangle take different colors.  Therefore in our span program the input is a randomly 3-colored graph and we look for a \emph{colorful} triangle and output its vertex that has color $0$. Having this vertex at hand, we can do Grover search among its neighbors and find two of them that are connected, as other vertices of the triangle, in time $O(n)$.

The span program is as follows:
\begin{itemize}
\item  The input vector space has dimension $4n$ and an orthonormal basis for it consists of vectors 
\begin{equation*}
 \big\{\ket{u, i}:\, u\in V, i \in \{0,1,2,3\}\big\}.
\end{equation*}

\item  $I_{\{u, v\}, q}$ is non-empty only when $q=1$ and the vertices $u, v$ have different colors. In this case, depending on their colors, we have
\begin{itemize}
\item[-] if $c(u)=0,c(v)=1$, then $I_{\{u, v\}, 1} =\big\{\ket{u, 0}-\ket{v, 1} \big\}.$
\item[-] if $c(u)=1,c(v)=2$, then $I_{\{u, v\}, 1} =\big\{\ket{u, 1}-\ket{v, 2} \big\}.$
\item[-] if $c(u)=2,c(v)=0$, then $I_{\{u, v\}, 1}=\big\{\ket{u, 2}-\ket{v ,3} \big\}.$
\end{itemize}

\item  Target vectors are  $\ket{t_v}:=\ket{v, 0}-\ket{v, 3}$ for all $v\in c^{-1}(0)$, meaning that if $\ket{t_v}$ is in the span of the available vectors, then $v$ is a vertex of the triangle with $c(v)=0$.

\end{itemize}

Consider an input graph that contains a unique triangle on vertices $x, y, z$ colored $0,1,2,$ respectively. Then the target vector $\ket{t_x}$ can be written as a linear combination of available vectors:
\begin{equation*}
\big(\ket{x, 0}-\ket{y, 1}\big)+\big(\ket{y, 1}-\ket{z, 2}\big)+\big(\ket{z, 2}-\ket{x, 3}\big)=\ket{x, 0}-\ket{x, 3}=\ket{t_x}.
\end{equation*}
Thus, the positive witness size equals $O(1)$.

Construction of a negative witness $\ket{\bar w_G}$ needs more work. 
We first construct a graph $H$ out of $G$ as follows. Vertices of $H$ are the same as those of $G$ except that vertices with color $0$ are doubled. In this regard, we denote a vertex $v$ of $G$ with color $c(v)=i\in \{1, 2\}$ by $v_i$ in $H$. Moreover, a vertex $u$ with color $0$ has two copies $u_0$ and $u_3$ in $H$.  More explicitly,
$$V(H) = \{v_i: \forall v\in V, c(v)=i\in \{1,2\}\}\cup \{v_0, v_1:\, v\in V, c(v)=0\}.$$
The edges of $H$ are described as follows. First of all, there is no edge between two vertices of $H$ that have the same color. Second, 
the edges of $H$ between vertices with colors $1$ and $2$ remain the same as in $G$.  Third, the adjacent vertices to the doubled vertices are as follows: if $c(u)=0$ then $u_0$ is connected to $u_3$ in $H$, and to every vertex $v_1$ with $v$ being connected to $u$ in $G$; also, $u_3$ is connected (to $u_0$ and) to every vertex $w_2$ with $w$ being connected to $u$ in $G$. This completes the description of our new graph $H$.

The graph $H$ has exactly one cycle, namely $x_0-y_1-z_2-x_3-x_0$. If we \emph{contract} this cycle to a new vertex $A$, we get an acyclic graph $H'$ consisting of a union of trees. We fix a vertex $r$ as a root in every connected component of $H'$ as follows. We first fix a vertex $s$ of $G$, and in every subtree of $H'$ we let its root $r$ be the vertex with the minimum distance from $s$ in $G$.\footnote{If this vertex is not unique, let $r$ be the least one in some predetermined order} Now we assign a number $\gamma$ to every vertex. For any subtree with root $r$, we let $\gamma(r)=0$. Then we traverse the subtree from its root to its leaves. We assign the same number to $v_i$ as the previous vertex except when we move from a node labeled $u_0$ to $u_3$ or vice versa.  In the later cases,  when we move from $u_0$ to $u_3$ we decrease the assigned number by $1$, and when we move from $u_3$ to $u_0$ we  increase  the assigned number by $1$.
Finally, we define $\gamma(x_0)=\gamma(x_3)=\gamma(y_1)=\gamma(z_2)=\gamma(A)$. Now define
\begin{equation*}
\ket{\bar{w}_G}=\sum_{u} \gamma(u_i)\ket{u, i}.
\end{equation*}
It is not hard to verify that $\ket{\bar w_G}$ is orthogonal to all available vectors and  $\langle \bar w_G| t_v\rangle=1$ for all $v\neq x$ with $c(v)=0$.
Thus $\ket{\bar w_G}$ is a valid negative witness. The size of this negative  witness is upper bounded by  
$$\sum_{u, v} (\gamma(u_i) - \gamma(v_j))^2\leq 4n \sum_{u} |\gamma(u_i)|^2.$$
This is of order $O(n^4)$ in the worst case (over the choice of the coloring) since in general we have $-n\leq\gamma(u_i)\leq n$. Nevertheless, in most cases $\gamma(u_i)$'s are small numbers. Indeed, for every vertex $u_i$, by definition $\gamma(u_i)$ is at most $h(u_i)$, the depth of the vertex $u_i$ in the associated subtree of $H'$ (its distance from the root).  On the other hand, as mentioned above, in $H$ and in $H'$ we remove all edges of $G$ between vertices with the same color. Thus each edge of $G$ will be removed with probability $1/3$. Therefore, 
 the \emph{expected} size of the negative witness over the random choice of the coloring is upper bounded as
\begin{align*}
\mathbb E\left[4n \sum_{u} |\gamma(u_i)|^2\right] &\leq \mathbb E\left[4n \sum_{u} |h(u_i)|^2\right] \\
&\leq 4n\sum_{u} \sum_{k=1}^\infty k^2  \left(\frac23\right)^k\\
&=O(n^2).
\end{align*}
Here the second line follows from the fact that if $h(u_i)=k$, then the first $k$ edges of $G$ in the path from $u_i$ to $s$ must be present in $H'$.
Therefore, the expected complexity of this span program is $O(\sqrt{n^2\cdot 1})=O(n)$, and by Markov's inequality  with a constant probability over the random choice of the coloring, the complexity of the span program is $O(n)$.
\end{proof}

\section{Span Program for Relations}\label{sec:GQQC}
The methods we discussed so far  enable us to come up with quantum query algorithms for  functions.  A question at this stage is how much we can generalize these methods to deal with relations instead of functions.  
A natural approach to solve the relation evaluation problem is to use the \textit{state conversion} problem introduced by Lee et al.\ \cite{LMRSS11} to study the function evaluation problem.  
In the state conversion problem we are given query access to an input $x\in [\ell]^n$ and we are asked to convert an initial state $\rho_x$ to a final state $\sigma_x$ using as few queries to the input oracle as possible. The adversary bound has been generalized for these problems as well. Relation evaluation can be seen as a special case of the state conversion problem. Using this idea, Belovs~\cite{Bel15} gave a tight lower bound on the quantum query complexity of evaluating relations with bounded error.  

In this section we try to generalize the method of non-binary span program for relations. We will introduce span programs for certain relation evaluation problems and using results of~\cite{Bel15} show that they provide upper bounds on their quantum query complexity.

A relation $r$ can be thought of as  a function $r:D_r \to 2^{[m]}$ from $D_r\subseteq [\ell]^n$ to subsets of $[m]$, in which $x\in D_r$ is in relation with all elements of $r(x)$.
Given $x\in D_r$ by evaluation of such a relation for we mean to output some $\alpha\in r(x)$.

In the following we assume that
for every $x\in D_r$, $|r(x)|=k$ is a constant independent of $x$ (e.g., for functions we have $k=1$). Our span program for relations work only with this extra assumption. 
A span program $P$ for such a relation $r: D_r \to 2^{[m]}$ in the sense of NBSPwOI consists of
\begin{itemize}
\item  a finite-dimensional inner product space $V$

\item $m$ non-zero vectors $|t_0 \rangle, |t_2 \rangle,\ldots ,|t_{m-1} \rangle\in V$

\item input vector sets $I_{j,q}\subseteq V$ for all $1\leq j\leq n$ and $q\in [\ell]$.

\end{itemize}
We then define the set $I\subseteq V$ by 
$$I=\bigcup_{j=1}^n \bigcup_{q\in [\ell]} I_{j,q},$$
and as before  the set of available vectors by 
$$I(x)=\bigcup_{j\in[n]}I_{j,x_j}.$$
Also, the matrix $A$ of size $d\times |I|$, where $d=\dim V$, is defined as before.

We say that $P$ evaluates the relation $r$ if for each $x\in D_r$ the vector $\sum_{\alpha\in r(x)} \ket{t_\alpha}$ belongs to the span of the available vectors $I(x)$, and $\ket{t_{\beta}}$ for $\beta\notin r(x)$ does not belong to the span of $I(x)$. Even more, there should be two witnesses indicating these. Namely, there must exist a positive witness  $\ket{w_x} \in \mathbb C^{|I|}$ and a negative witness $\ket{\bar{w}_x}\in \mathbb C^d$ satisfying the followings: 
\begin{itemize}
\item First, the coordinates of $\ket{w_x}$ associated to unavailable vectors are zero. 
\item Second, $A\ket{w_x} = \frac{1}{|r(x)|}\sum_{\alpha\in r(x)} \ket{t_\alpha}$. 
\item Third, for all $\ket v\in I(x)$ we have $\langle v|\bar w_x\rangle =0$. 
\item Fourth, for all $\beta\notin r(x)$ we have $\langle t_{\beta}|\bar{w}_x\rangle=1$. 
\end{itemize}
The triple of the span program together with the set of positive and negative witnesses is denoted by $(P, w, \bar w)$.

Now we define the complexity of $(P, w, \bar w)$ for a relation $r$ similar to that of a function. For every $x\in D_r$ we define
$$\mathrm{wsize}_x(P,w_x,\bar{w}_x):=\max\big\{\|  \ket{w_x}\| ^2,\|  A^\dagger \ket{\bar{w}_x}\| ^2\big\}.$$
Next the complexity of $(P, w, \bar w)$ equals
\begin{align*}
 \mathrm{wsize}(P,w,\bar{w}):= \max_{x\in D_f} \; \mathrm{wsize}_x(P,w_x,\bar{w}_x). 
\end{align*}

Similar to the span program for functions we say that a span program $(P,w,\bar{w})$ for a relation $r$ is canonical if
\begin{itemize}
\item The underlying vector space is $|D_r|$-dimensional ($d=|D_r|$), and an orthonormal basis for this vector space is $\{|e_x\rangle:\,   x\in D_r\}$.

\item For any $\alpha\in [m]$ the target vector is $
\ket{t_\alpha}=\sum_{x: \alpha\notin r(x)}\ket{e_x}$
\item For any $x\in D_r$ the negative witness $\ket{\bar{w}_x}$ equals $\ket{e_x}$. As a consequence for all $ |v\rangle \in I(x)$ we have $ \langle e_x |v\rangle =0$.
\end{itemize}
As in the case of non-binary span programs, a span program $(P,w,\bar{w})$ for a relation $r:D_r \to 2^{[m]}$ has an equivalent canonical span program $(P',w',\bar{w}')$  with the same complexity.  The proof of this fact is similar to that of Proposition~\ref{pro:SpanProgram2can} 
and is not repeated here. 

Belovs  has shown in Theorem~40 of~\cite{Bel15} that assuming that given $\alpha\in[m]$ and $x\in D_r$, we can \emph{efficiently verify} whether $\alpha\in r(x)$ or not, then the following optimization program gives a tight lower bound for the quantum query complexity of evaluating the relation $r$ with bounded error (meaning that with high probability we can sample from elements of $r(x)$ given query access to $x$):

\begin{subequations}\label{SDP:AdvRelation}
\begin{align}
&\text{minimize} && \max\Big\{\max_{x\in D_r}\sum_{j=1}^n \|  \ket{u_{x,j}} \| ^2,\max_{x\in D_r}\sum_{j=1}^n\|  \ket{v_{x,j}} \| ^2\Big\} \\
&\text{subject to} && 1-\sum_{\substack{\alpha\in[m]}}\langle\sigma_{x,\alpha}|\sigma_{y,\alpha} \rangle = \sum_{\substack{1\leq j\leq n\\ x_j\neq y_j }} \langle u_{x,j}| v_{y,j} \rangle   \qquad \text{for all } x,y\in D_r;  \\ 
& &&\sum_{\substack{\alpha \notin r(x)}} \||\sigma_{x,\alpha}\rangle\|^2 = 0 \hspace{113pt} \text{for all } x\in D_r. 
\end{align}
\end{subequations}
Observe that in the first constraint, by putting $x=y$ we obtain 
$\sum_\alpha \|\ket{\sigma_{x, \alpha}}\|^2=1$. Thus letting 
$$\ket{\sigma_x}:= \bigoplus_\alpha \ket{\sigma_{x, \alpha}},$$
$\ket{\sigma_{x}}$ would be a normalized vector and a quantum state. Thus the above optimization problem is nothing but  the dual adversary bound for the state conversion problem~\cite{LMRSS11} of converting a fixed state to $\ket{\sigma_{x}}$ on which we optimize with a given constraint (the second one).

We can now state our main result about relations.

\begin{theorem} \label{thm:RelSpanProgram2SDP}
Suppose that $r:D_r \to 2^{[m]}$ with $D_r\subseteq [\ell]^n$ is a relation such that $|r(x)|=k$ for some constant $k$ independent of $x$. Let $(P,w,\bar{w})$ be a span program evaluating $r$ as defined above, with complexity $C$. Then the quantum query complexity of $r$ is at most $C$.
\end{theorem}

\begin{proof}
The proof is very similar to the proof of Theorem~\ref{thm:NBSP2ADV}; we convert the given span program to a feasible solution of the dual adversary SDP~\eqref{SDP:AdvRelation} that characterizes the quantum query complexity of $r$. Without loss of generality we assume that $(P,w,\bar{w})$ is canonical. Then for each $x\in D_r$ we have 
\begin{align*}
A\ket{w_x}=\frac1k\sum_{\alpha\in r(x)}\ket{t_\alpha},
\end{align*}
and then for any $y\in D_r$
\begin{align*}
\langle e_y|A\ket{w_x} & =  \frac 1k \sum_{\alpha\in r(x)} \langle e_y| t_\alpha\rangle = \frac 1 k \big| r(x)\setminus r(y)\big| = 1- \frac 1 k \big| r(x)\cap r(y) \big|.
\end{align*}
Now as in the proof of Theorem~\ref{thm:NBSP2ADV} we may define vectors $\ket{u_{x, j}}$ and $\ket{v_{x, j}}$ from rows of $A$ and the positive witnesses. So let
\begin{align*}
|u_{x,j}\rangle=\bigoplus_{q\in[\ell]} |a_{x,jq}\rangle,
\qquad
|v_{x,j}\rangle=\bigoplus_{q\in[\ell]} |w_{x,jq}\rangle.
\end{align*}
Then the above equation says that 
$$\sum_{j: x_j\neq y_j}\langle u_{x,j}|v_{y, j}\rangle= 1-\frac1k|r(x)\cap r(y)|.
$$
We also set
\begin{align}
 \ket{\sigma_{x,\alpha}}=
\begin{cases} 
 \frac{1}{\sqrt{k}}\ket{e_\alpha}  & \alpha\in r(x),\\
 0 & \text{ otherwise}.
 \end{cases}
\end{align}
These give a feasible solution for~\eqref{SDP:AdvRelation} with the same objective value as the complexity of $(P,w,\bar{w})$.
\end{proof}

\section{Non-binary Learning Graph}\label{sec:NBLG}

Learning graph is another computational model introduced by Belovs~\cite{Bel12} that is used for the design of quantum query algorithms for functions with binary output.
Learning graph somehow models the flow of information we obtain about the output of the function while we make queries to its input. Such queries are made one by one until we can 
\emph{certify} the output of the function. 
In this section we generalize the learning graph technique for finding quantum query algorithms for arbitrary functions $f:D_f \to [m]$ with $D_f\subseteq [\ell]^n$.

We first need a few definitions before introducing learning graphs. These are straightforward generalizations of the notions introduced by Belovs~\cite{Bel12, Bel14} to the non-binary case.  

\begin{definition} 
\begin{enumerate}\rm
\item\label{def:cer}
\textbf{(Certification)}  Let $f: D_f\to [m]$ be a function, $x\in f^{-1}(\alpha)$ be an input, and $S\subseteq \{1, \dots, n\}$ be a nonempty set of indices. Let $x_S$ be the substring of $x$ whose indices come from the subset $S$.  We say that $S$ certifies $f(x)=\alpha$, if for every $y\in D_f$ with $x_S=y_S$ we have $f(y)=f(x)=\alpha$.
\item \label{def:icer}
\textbf{($\alpha$-certificate)}  Given a function $f: D_f\to [m]$ and an input $x\in f^{-1}(\alpha)$, we say that 
$\mathcal{M}_x$, a collection of subsets of $\{1, \dots, n\}$, is an $\alpha$-certificate of $f$ for $x$, if (i) each $S\in \mathcal{M}_x$ certifies $f(x)=\alpha$ and (ii) $\mathcal{M}_x$ is closed under taking supersets, i.e.,
\begin{equation}\label{eq:icer}
\forall S,S': S\in  \mathcal{M}_x, S\subseteq S' \quad \Rightarrow \quad S'\in \mathcal{M}_x.
\end{equation} 
Indeed, if $S$ certifies $f(x)=\alpha$, so does any superset of $S$.
\item\label{def:cerStru}
\textbf{(Certificate Structure)}  A certificate structure is a collection $\mathcal{E}$ of certificates for all $x\in D_f$ with $f(x)\neq 0$. In other words, $\mathcal E$ is a certificate structure for $f$ if it contains an $\alpha$-certificate $\mathcal M_x$ for every $x$ with $f(x)=\alpha\neq 0$.
We emphasis that we do not need certificates for those $x$ with $f(x)=0$. 
\end{enumerate}
\end{definition}

A learning graph is designed based on a certificate structure and can be converted to a quantum query algorithm for any function having the same certificate structure.
The learning graph and flow for functions with binary output ($m=2$) is defined by Belovs~\cite{Bel12,Bel14} whose definition can easily be generalized for arbitrary functions.

\begin{definition}[Learning graph]\label{def:LG}\rm A learning graph $\mathcal{G}$ is a \emph{weighted} acyclic directed graph whose set of nodes is a subset of the power set of $\{1, \dots, n\}$, and whose edge set contains only directed edges of the form $S\to S\cup \{j\}$, where $S\subset \{1, \dots, n\}$ and $j\in \{1, \dots, n\}\setminus S$. The root of $\mathcal{G}$ is fixed to be the empty set $\emptyset$, and the weight of an edge $e$ is denoted by $w_e$.

Given a function $f: D_f\to [m]$ with a fixed certificate structure $\mathcal E$, a flow on a learning graph is a collection of flows on the graph for every $x$ with $f(x)\neq 0$ as follows:

\begin{itemize}
\item The value of the flow for $x\in D_f$ with $f(x)=\alpha\neq 0$ associated to the edge $e$ of the learning graph is denoted by $p_e(x)$.


\item The only \emph{source} of the flow is the vertex $\emptyset $.

\item A vertex $S$ is a \emph{sink} only if $S\in \mathcal{M}_x$ belongs to the $\alpha$-certificate of $x$ (in $\mathcal E$).

\item For each vertex $S$ which is not a sink nor a source, the sum of  $p_e(\mathcal{M}_x)$'s over all edges $e$ leaving  $S$ is equal to the sum of $p_e(\mathcal{M}_x)$'s over all incoming edges $e$ to $S$.

\item  The value of the flow is 1, meaning that the sum of all $p_e(\mathcal{M}_x)$ on edges leaving $\emptyset$ equals 1.
\end{itemize}
\end{definition}

Given a learning graph $\mathcal{G}$ together with a collection of flows $p_e(\mathcal{M}_x)$ for every $x$ with $f(x)\neq 0$ as above, the \emph{complexity} of the learning graph is defined as
\begin{equation}
\mathcal{C}(\mathcal G, p_e)=\sqrt{\mathcal{N}(\mathcal G, p_e)\cdot\mathcal{P}(\mathcal G, p_e)},
\end{equation}
where
\begin{equation}
\mathcal{N}(\mathcal G, p_e)=\sum_{e} w_e \hspace{30pt}\text{and}\hspace{30pt} \mathcal{P}(\mathcal G, p_e)=\max_{x: f(x)\neq 0}\sum_{e} \frac{p_e^2(\mathcal{M}_x)}{w_e}.
\end{equation}

We can now state the main theorem of this section.

 \begin{theorem}\label{thm:LGraph}
Let $(\mathcal{G}, p_e)$ be a learning graph together with a collection of flows as above for a function $f: D_f\to [m]$ with  $D_f\subseteq [\ell]^n$. Then there is a solution to the dual adversary SDP for the same function with objective value being $O(\mathcal{C}(\mathcal G, p_e))$. Thus $O(\mathcal{C}(\mathcal G, p_e))$ is an upper bound on the quantum query complexity of $f$.
\end{theorem}

The definition of the learning graph with flows and their complexity for functions with binary output $(m=2)$ have been proposed in~\cite{Bel14}, which we easily generalized for arbitrary functions. Moreover, the above theorem has been proven in~\cite{Bel14} in the case of $m=2$.  Below we will give a proof that works for arbitrary functions.

There are indeed two proofs of Theorem~\ref{thm:LGraph} when $m=2$. The first one, which works only for $\ell=2$, is based on converting a learning graph $(\mathcal G, p_e)$ to a span program for $f$ with the same complexity~\cite{Bel12}. The second proof, is based on converting the learning graph directly to a feasible solution of the dual adversary SDP~\cite{LB11}. Here for non-binary learning graphs we generalize the former proof and find a feasible solution to the dual adversary SDP using non-binary span programs.


\begin{proof}
Let $\tilde f: D_f\to \{0,1\}$ be the binary version of $f$ which decides whether $f(x)$ is zero or not, i.e., $\tilde f(x)$ equals $0$ if $f(x)=0$, and equals $1$ if $f(x)\neq 0$. Observe that $(\mathcal G, p_e)$ is a valid learning graph for $\tilde f$ as well. Then by the special case of the theorem for $m=2$ (which has been proven in~\cite{Bel14}), with $O(\mathcal{C}(\mathcal G, p_e))$ queries to $x$ we can decide $\tilde f(x)$. If $\tilde f(x)=0$ the value of $f(x)$ would be determined. Otherwise we know that $f(x)\neq 0$ and we must determine $\alpha=f(x)\in \{1, \dots, m-1\}$.

By the above discussion, with no loss of generality we may restrict ourself to $D'_f= D_f\setminus f^{-1}(0)$. In other words, we may assume that there is no $x$ with $f(x)=0$.  Note that in this case there is a flow $p_e(x)$ for every $x\in D'_f$.

To prove this theorem we first build an NBSPwOI having complexity equal to that of the learning graph up to a factor of $\sqrt{\ell-1}$. Then we remove the extra factor of $\sqrt{\ell-1}$ by converting it to an NBSP. 

As we mentioned before, vertices of the learning graph are subsets $S$ of $\{1,\dots, n\}$, an edge $S\to S\cup \{j\}$ of the learning graph is denoted by $e_{S, j}$, and by $\psi_S$ we mean an assignment of elements of $S$, i.e., $\psi_S\in [\ell]^S$.
The non-binary span program with orthogonal inputs is as follows:
\begin{itemize}
\item An orthonormal basis of the vector space $V$ is given by
\begin{equation*}
\big\{\ket{S, \psi_S, \alpha} : S\subseteq \{1, \dots, n\}, \psi_S \in [\ell]^{S}, \alpha \in [m]\big\}\cup\{\ket{t_\alpha}:\alpha\in[m]\}.
\end{equation*}
\item The set $I_{j,q}$ consists of all vectors of the form
\begin{align}\label{eq:ijq-alpha}
\sqrt{w_{e_{S, j}}} \big(-\ket{S,\psi_S}+\ket{S\cup\{j\},\psi_S\cup\{j\to q\}}\big)\otimes \ket{\alpha},
\end{align}
where $ e_{S, j}$ is an edge of the learning graph, $\psi_S\in [\ell]^{S}$ is arbitrary and $\psi_S\cup\{j\to q\}$ means an assignment of $S\cup\{j\}$ which coincides with $\psi_S$ on $S$ and assigns $q$ to the $j$-th index. Finally $\alpha\in[m]$ is arbitrary. 

\item Target vectors are $\{\ket{t_\alpha}: \alpha\in[m]\}$.

\item Free input vectors are 
$$I_{\text{free}}=\Big\{\ket{t_\alpha}-\big(\ket{\emptyset} -\ket{S,\psi_S}\big)\otimes \ket{\mu_\alpha} : ~ \alpha\in[m] , (S,\psi_S) \text{ is an $\alpha$-certificate}\Big\},$$
where $\ket{\mu_\alpha}$'s are vectors defined in~\eqref{eq:mu}, and by $\ket{\emptyset}$ we mean the vector associated to the empty set (with the empty assignment).
\end{itemize}

Then for every $x\in f^{-1}(\alpha)$ we have 
\begin{align*}
\ket{t_\alpha}=&-\sum_{e_{S,j}} \frac{p_e(x)}{\sqrt{w_{e_{S, j}}}}\Big(\sqrt{w_{e_{S, j}}}\big(-\ket{S,x_S}+\ket{S\cup \{j\},x_S \cup\{j\to x_j\}}\big)\Big)\otimes \ket{\mu_\alpha}\\
&+\sum_{(S,\psi_S)\text{: $\alpha$-certificate}}p_e(x)\Big(\ket{t_\alpha}-\big(\ket{\emptyset}-\ket{S,x_S} \big)\otimes \ket{\mu_\alpha}\Big)
\end{align*}
where $x_S\in [\ell]^S$ is the restriction of $x$ on $S$.
Thus the positive witness size for $x$  equals  $\| \ket{\mu_\alpha}\|^2\sum_e p_e^2(x)/w_e$.

For the negative witness define
\begin{equation}\label{eq:LGwbar}
\ket{\bar{w}_x}=\sum_{\beta:\,\beta\neq \alpha}\ket{t_{\beta}}+ \sum_{S}\ket{S, x_S} \otimes \ket{\nu_{\alpha}},
\end{equation}
where $\ket{\nu_\alpha}$ is defined in~\eqref{eq:nu}. By construction  $\ket{\bar{w}_x}$ is orthogonal to all free  and available input vectors. Moreover, we have $\langle \bar{w}_s \ket {t_\beta}=1$ for all $\beta\neq \alpha$. Thus $\ket{\bar w_x}$ is a valid negative witness. 
To calculate the negative witness size note that for each edge $e_{S, j}$ of the learning graph there are exactly $\ell-1$ input vectors that contribute to the negative witness size. So the negative witness size is $(\ell-1)\|\ket{\nu_\alpha}\|^2\sum_e w_e$. We conclude that the complexity of the span program is $\sqrt{\ell-1}$ times the complexity of the learning graph.

Now to finish the proof we remove the extra $\sqrt{\ell-1}$ factor by converting the above NBSPwOI to an NBSP. The idea is that the vectors in $I_{j,q}$ given by~\eqref{eq:ijq-alpha} for different values of $q$ are not orthogonal. Then it is natural to rewrite the above NBSPwOI as a general NBSP.
The resulting non-binary span program consists of
\begin{itemize}
\item $H=\bigoplus_j H_j\oplus H_{\rm free},$
\item An orthonormal basis for $H_j$ consists of 
$$\big\{\ket{j,S,\psi_S,\alpha}:~ e_{S,j}\text{ is an edge}, \psi_S\in [\ell]^S ,\alpha\in [m]\big\},$$
\item $H_{j,q}={\rm span}\big\{\big(-\ket{j,S,\psi_S}+\ket{j,S\cup\{j\},\psi_S\cup\{j\to q\}}\big)\ket{\alpha}:~  e_{S,j}\text{ is an edge},  \psi_S\in[\ell]^S, \alpha\in [m] \big\},$
\item An orthonormal basis for $H_{\rm free}$ consists of  
$$\{\ket{f_{\alpha,S,\psi_S}} : ~ \alpha\in[m] , (S,\psi_S) \text{ is an $\alpha$-certificate}\},$$
\item An orthonormal basis for $V$ is given by
$$\big\{\ket{S,\psi_S,\alpha}:~ S\subseteq\{1,\ldots,n\}, \psi_S\in[\ell]^S,\alpha\in[m]\big\}\cup\big\{\ket{t_\alpha}\big|\alpha\in [m]\big\},$$
\item $A:H\to V$ is given by
$$A=\sum
_{j,S,\psi_S,\alpha} \sqrt{w_{e_{S,j}}}\ket{S,\psi_S,\alpha}\bra{j,S,\psi_S,\alpha}+\sum_{\alpha,S,\psi_S}\ket{g_{\alpha,S,\psi_S}}\bra{f_{\alpha,S,\psi_S}},$$
where $\ket{g_{\alpha, S,\psi_S}}=\ket{t_\alpha}-\big(\ket{\emptyset}-\ket{S,\psi_S}\big)\ket{\mu_\alpha}$, 
\item target vectors are $\{\ket{t_\alpha}: \alpha\in[m]\}.$
\end{itemize}

Then for any $x\in f^{-1}(\alpha)$, the positive witness $\ket{w_x}\in H(x)$ is
\begin{align*}
\ket{w_x}=&-\sum_{e_{S,j}} \frac{p_e(x)}{\sqrt{w_{e_{S, j}}}}\Big(\big(-\ket{j,S,x_S}+\ket{j,S\cup \{j\},x_S \cup\{j\to x_j\}}\big)\Big)\otimes \ket{\mu_\alpha}\\
&+\sum_{(S,\psi_S)\text{: $\alpha$-certificate}}p_e(x)\ket{f_{\alpha,S,\psi_S}}.
\end{align*}
Therefore, the size of the positive witness equals 
$$2\| \ket{\mu_\alpha}\|^2\sum_e p_e^2(x)/w_e\leq 4 \sum_e p_e^2(x)/w_e.$$
The negative witness $\ket{\bar{w}_x}\in V$ is equal to~\eqref{eq:LGwbar}, but the negative witness size is
\begin{equation*}
\| \bra{\bar{w_x}}A\|^2=\|\ket{\nu_\alpha}\|^2\sum_{e} w_{e} \leq 4\sum_{e} w_{e}. 
\end{equation*}
As a result the factor of $\sqrt{\ell-1}$ does not appear in the complexity of this non-binary span program. We are done.
\end{proof}

The advantage of the non-binary learning graph is that, if we have a learning graph for a binary function having certificate structure $C$, we can use the same  learning graph to bound the query complexity of any \emph{non-binary} function having the same certificate structure. See the following example for clarification.

\begin{example}
Let $G$ be a graph whose edges are colored  by $c$ different colors, in such a way that all monochromatic triangles have the same color.
Assume that we have query access to colors of the edges of $G$ (if there is no edge in a queried pair of vertices, the query output will be $0$). Then the quantum query complexity of detecting the color of a monochromatic triangle in $G$ is $O(n^{9/7})$.
\end{example}
\begin{proof}
The best learning graph that has been defined for triangle finding problem in~\cite{LMS17} has complexity  $O(n^{9/7})$. Since the certificate structure of the triangle finding problem is the same as the problem of finding a monochromatic triangle in a colored graph, the learning graph of the former problem together with its flow is a valid learning graph for the latter. Therefore the quantum query complexity of detecting the color of a monochromatic triangle in a graph, whose edges are colored arbitrarily using $m$ different colors, is $O(n^{9/7})$.
\end{proof}

\section{Concluding Remarks}\label{sec:conclusion}
In this paper we generalized the notion of span program for non-binary functions. We showed that our non-binary span program gives a characterization of the quantum query complexity of non-binary functions that is tight up to a constant factor. We also introduced NBSPwOI as a special case of non-binary span programs, which although is an intuitive tool to help us design non-binary span programs, has an unavoidable factor of $\sqrt{\ell-1}$ in its complexity. Yet, when $\ell$ is a constant, in designing quantum query algorithms we may restrict ourselves to non-binary span programs with orthogonal inputs. Moreover, as shown in some examples, NBSPwOIs can be used as an intermediate step in designing NBSPs and to find tight bounds on the quantum query complexity of non-binary functions.


Span programs have been used to design quantum query algorithms for various problems such as formula evaluation~\cite{Rei09}, the rank problem~\cite{Bel11}, $st$-connectivity and claw detection~\cite{BR12}, graph collision problem~\cite{ABIOS13}, tree detection~\cite{Wang13}, graph bipartiteness and connectivity~\cite{Ari15}, and detecting cycles~\cite{CMB16}. Now that we have a natural generalization of span program for  functions  with non-binary input/output alphabets, we can use ideas from these problems to design new quantum query algorithms for function with non-binary input/output alphabets, as what we did for the problem of triangle finding. Besides this example, we already have other applications of the non-binary span program, especially for simplification of the proof and generalization of some existing results for binary functions which will be presented in future works.

We also suggested a method for proving bounds on the quantum query complexity of non-binary functions via learning graphs. In the binary case, the method of
learning graphs has been used to design quantum query algorithms for graph collision, $k$-distinctness~\cite{Bel2012} and triangle finding problems~\cite{LMS17}.  Now a natural question is, can we generalize the ideas behind these learning graphs to prove upper bounds on the quantum query complexity of some non-binary functions?

\section*{Acknowledgements} 
The authors are thankful to Ronald de Wolf for his comments on an early version of this paper and to Stacey Jeffery for elaborations on her results in~\cite{Jeffery14, ItoJeffery15}.  They are also grateful to unknown referees whose suggestions and comments significantly improved the presentation of the paper.

\bibliography{references}
\bibliographystyle{alpha}

\appendix
\section{Zero-error Quantum Algorithms vs Non-binary Span Program} \label{app:qa2non-sp}

For every \emph{one-sided error} quantum query algorithm evaluating a binary function $f:\{0,1\}^n\to\{0,1\}$ using $Q$ queries to the input oracle, there exists a span program computing the same function having complexity $Q$ up to a constant factor~\cite{Rei09}. Here we generalize this result to the non-binary case, showing an equivalence between \emph{zero-error} quantum query algorithms for arbitrary functions and non-binary span program. The reason that we consider zero-error algorithms instead of one-sided error ones here, is that for functions with non-binary output one-sided error algorithms is not well-posed.

\begin{theorem}\label{thm:0alg2sp}
For every zero-error quantum query algorithm  evaluating a function $f:D_f\to[m]$, with $D_f\subseteq [\ell]^n$, using $Q$ queries to input oracle, there exists a non-binary span program computing the same function with complexity $O(Q)$.
\end{theorem}

\begin{proof}
Let us assume that the algorithm is performed with three registers, the first one as the workspace, the second one for storing the index $1\leq j\leq n$ and third one for storing the output. We also assume that after finding the answer, the algorithm copies it to the output register and then un-computes the first two registers by running the algorithm in reverse; this doubles the number of queries to the input oracle.

Suppose that the algorithm starts at the quantum state $\ket{\phi_0}=\ket{0^r,1, 0}$ and finishes with
\begin{equation*}
\ket{\phi_{2Q+1}}=U_{2Q+1}O_x \ldots U_{3}O_x U_1 \ket{0^r,1, 0}.
\end{equation*}
Here $U_{2\tau+1}$ for $\tau\in\{0,\ldots,Q\}$ is a unitary independent of input $x$, and $O_x$ is the oracle given by 
$O_x|p,j, \alpha \rangle=\omega^{x_j}|p,j, \alpha\rangle$
where $\omega=e^{2\pi i/\ell}.$
By the above assumption on the structure of the algorithm and the fact that the algorithm has no error we have
$$\ket{\phi_{2Q+1}} =\ket{0^r, 1, f(x)}.$$
Let $\ket{\phi_\tau}$ be the quantum state at $\tau$-th step of the algorithm:
\begin{equation*} \ket{\phi_\tau}=:
\begin{cases}
 U_\tau\ket{\phi_{\tau-1}} 
 &\tau \textrm{ is odd}\\
 O_x\ket{\phi_{\tau-1}}
 & \tau\textrm{ is even}.
 \end{cases}
\end{equation*}
Then we may define the span program as follows:
\begin{itemize}
\item The inner product space is $(2q+2)mn 2^r$-dimensional with orthogonal basis:
$$\Big\{ \ket{\tau,p,j, \alpha}: \tau\in\{0,1,\ldots,2q+1\}, p\in \{0,1\}^r, j\in\{1,\ldots,n\}, \alpha\in [m]\Big\} $$
\item For every $\alpha\in [m]$ the target vector is $\ket{t_\alpha}=\ket{0,0^r,1, 0}-\ket{2Q+1,0^r,1, \alpha}.$
\item Free input vectors are:
$$ \ket{v_{\tau,p,j, \alpha}}:=
\ket{\tau-1,p,j, \alpha}-U_{\tau}\ket{\tau,p,j, \alpha}, $$
for every odd $\tau$ and arbitrary   $p, j, \alpha$.
These vectors are always available no matter what the input is.
\item $I_{j,q}=\Big\{\ket{v^q_{\tau,p,j, \alpha}}:=\ket{\tau-1,p,j, \alpha}-\omega^q\ket{\tau,p,j, \alpha}: \tau \text{ even}, p\in\{0,1\}^r, \alpha\in[m] \Big\}$
\end{itemize}

For a given $x\in D_f$ and even $\tau$ let us define
$$\ket{v_{\tau,p,j, \alpha}} := \ket{v_{\tau,p,j, \alpha}^{x_j}} =\ket{\tau-1,p,j, \alpha}-\omega^{x_j}\ket{\tau,p,j, \alpha}=\ket{\tau-1,p,j, \alpha}-O_x\ket{\tau,p,j, \alpha}.$$
Then $I(x) = \big\{\ket{v_{\tau,p,j, \alpha}} : 1\leq \tau\leq 2Q+1, p\in \{0,1\}^r, 1\leq j\leq n, \alpha\in [m]\big\}$.
Consider the linear combination of these vectors with coefficients $\bra{p, j, \alpha} \phi_{\tau-1}\rangle$ as follows:
\begin{align*}
\sum_{\tau=1}^{2Q+1}\sum_{p, j, \alpha} &\bra{p,j, \alpha}\phi_{\tau-1}\rangle \ket{v_{\tau,p,j, \alpha}} \\
& 
=\sum_{\tau: \text{ odd}} \ket{\tau-1}\ket{\phi_{\tau-1}}-U_{\tau}\ket{\tau}\ket{\phi_{\tau-1}}
 + \sum_{\tau: \text{ even}} \ket{\tau-1}\ket{\phi_{\tau-1}}-O_x\ket{\tau}\ket{\phi_{\tau-1}} \\
& = \sum_{\tau=1}^{2Q+1}  \ket{\tau-1}\ket{\phi_{\tau-1}} - \ket{\tau}\ket{\phi_{\tau}}\\
&=\ket{0}\ket{\phi_0}-\ket{2Q+1}\ket{\phi_{2Q+1}} 
\\&=\ket{t_{f(x)}}.
\end{align*}
Therefore, $\ket{t_{f(x)}}$ belongs to the span of $I(x)$, and the positive witness size equals
$\sum_{\tau \text{: even}} \parallel \ket{\phi_{\tau-1}}\parallel^2=Q$. Here we do not count the coefficients of $\ket{v_{\tau, p, j, \alpha}}$ for odd $\tau$'s since they correspond to free vectors. 

For the negative witness define 
$$\ket{\bar{w}_x}=\sum_{\tau=0}^{2Q+1} \ket{\tau}\ket{\phi_\tau}.$$ 
Observe that $\ket{\bar{w}_x}$ is a negative witness since
\begin{itemize}
\item $\forall \alpha\neq f(x):\, \bra{t_\alpha}\bar{w}_x\rangle= 1-\bra{0^r,1, \alpha}\phi_{2Q+1}\rangle = 1.$
\item $\forall \ket{v_{\tau, p, j, \alpha}^{x_j}}\in I_{j,x_j}:$ $\langle v_{\tau, p, j, \alpha}^{x_j}\ket{\bar{w}_x}= \bra{p,j, \alpha}\phi_{\tau-1}\rangle-  \bra{p, j , \alpha} O_x^{\dagger} | \phi_\tau \rangle=0.$
\item $\forall \ket{v_{\tau, p, j, \alpha}}\in I_{\text{free}}:$
$\bra{v_{\tau,p,j}}\bar{w}_x\rangle=\bra{p,j, \alpha}\phi_{\tau-1}\rangle-  \bra{p, j , \alpha} U_\tau^{\dagger} | \phi_\tau \rangle=0$.  
\end{itemize}
The negative witness size equals 
\begin{align*}
 \sum_{\tau: \text{ even}}\, \sum_{p,j, \alpha}\,\sum_{q \in [\ell]\setminus\{x_j\} }
\Big|\bra{\bar{w}_x} v^q_{\tau, p, j, \alpha}\rangle \Big|^2 & =\sum_{\tau: \text{ even}}\, \sum_{p,j, \alpha}\,\sum_{q \in [\ell]\setminus\{x_j\} }
\Big|\bra{\bar{w}_x} {\tau-1,p,j, \alpha}\rangle- \omega^q \bra{\bar{w}_x}{\tau,p,j, \alpha}\rangle \Big|^2\\&=
\sum_{\tau: \text{ even}}\, \sum_{p,j, \alpha}\,\sum_{q \in [\ell]\setminus\{x_j\} }\Big| \bra{\phi_{\tau-1}}p,j, \alpha\rangle-\omega^{q} \bra{\phi_\tau}p,j, \alpha\rangle\Big|^2 \\&=
 \sum_{\tau: \text{ even}} \,\sum_{p,j, \alpha}\,\sum_{q \in [\ell]\setminus\{x_j\} } \big|(1-\omega^{q+x_j}) \bra{\phi_{\tau-1}}p,j, \alpha\rangle \big|^2\\&=
 \sum_{\tau: \text{ even}} \,\sum_{p,j, \alpha} \big|\bra{\phi_{\tau-1}}p,j, \alpha\rangle \big|^2 \Big(\sum_{q \in [\ell]\setminus\{x_j\} } \big|(1-\omega^{q+x_j})  \big|^2\Big)
\\&\leq
4Q(\ell-1). 
\end{align*}
Therefore, the complexity of this span program is $O(\sqrt{\ell-1}\,Q)$.
Now we convert this NBSPwOI into a non-binary span program to eliminate the factor of $\sqrt{\ell-1}$ in its complexity. The resulting non-binary span program consists of:
\begin{itemize}
\item $H=\bigoplus_j H_j\oplus H_{\rm free},$ 
\item An orthonormal basis for $H_j$ is given by $\big\{\ket{j}\otimes  \ket{\tau,p,j,\alpha}:\tau\in[2q+1],p\in\{0,1\}^r,j\in\{1,\ldots,n\},\alpha\in [m]\big\},$
\item $H_{j,q}={\rm span}\big\{\ket{j}\otimes \ket{v^q_{\tau,p,j, \alpha}}:\tau \text{ even}, p\in\{0,1\}^r, \alpha\in [m]\big\},$ where as before $\ket{v^q_{\tau,p,j, \alpha}}:=\ket{\tau-1,p,j, \alpha}-\omega^q\ket{\tau,p,j, \alpha}$,
\item An orthonormal basis for $H_{\rm free}$ is given by $\big\{\ket{f_{\tau,p,j, \alpha}}: \tau \text{  odd},  p\in\{0,1\}^r, j\in\{1,\ldots,n\}, \alpha\in [m]\big\},$
\item  An orthonormal basis for $V$  is given by $\big\{ \ket{\tau,p,j,\alpha}:\tau\in[2q+1],p\in\{0,1\}^r,j\in\{1,\ldots,n\},\alpha\in [m]\big\},$
\item $A:H\to V$ is given by  
$$A=\sum_{j,q,\tau\text{ even},p,\alpha }\ket{\tau,p,j, \alpha}\bra{j}\otimes \bra{\tau,p,j, \alpha}
+\sum_{j,q, r\text{ odd},p,\alpha}\ket{v_{\tau,p,j, \alpha}}\bra{f_{\tau,p,j, \alpha}},$$
\item $\ket{t_\alpha}=\ket{0,0^r,1, 0}-\ket{2Q+1,0^r,1, \alpha}$ for every $\alpha\in [m]$.
\end{itemize}
For any $x\in f^{-1}(\alpha)$ the positive witness $\ket{w_x}\in H(x)$ is
$$\sum_{\tau \text{ even}}\sum_{p, j, \alpha}\bra{p,j, \alpha}\phi_{\tau-1}\rangle \ket{j}\ket{v_{\tau,p,j, \alpha}}+\sum_{\tau \text{ odd}}\sum_{p, j, \alpha}\bra{p,j, \alpha}\phi_{\tau-1}\rangle \ket{f_{\tau,p,j,\alpha}}.$$
So the positive witness size remains the same. The negative witness $\ket{w_x}\in V$ is also the same, but the negative witness size equals
\begin{equation*}
|\bra{\bar{w}_x}A\|^2=\|\sum_{\tau=0}^{2Q+1} \bra{\tau}\bra{\phi_\tau}A\|^2=
\sum_{j,\tau,p,\alpha }|\langle\phi_\tau\ket{p,j, \alpha}|^2=O(Q).
\end{equation*}
Therefore, the complexity of this NBSP is $O(Q)$.
\end{proof}


\section{Proof of Proposition~\ref{pro:SDP2SPwob}} \label{app:SDP2SPwob}

 Let $\ket{u_{x, j}}$'s and $\ket{v_{x, j}}$'s form a feasible solution of~\eqref{eq:dual-SDP} satisfying
$$\sum_{j: x_j\neq y_j}\langle u_{xj}|v_{yj}\rangle=1- \delta_{f(x), f(y)}.$$ 
We convert this feasible solution to a canonical non-binary span program with orthogonal inputs.
   Define vectors $\ket{a_{x, jq}}$ and $\ket{w_{x, jq}}$ by
$$\ket{a_{x,jq}}=(1-\delta_{x_j,q})\gamma^{-1}\,|u_{x,j}\rangle, 
\qquad\quad |w_{x,jq}\rangle=\delta_{x_j,q} \gamma\,|v_{x,j}\rangle,$$
where $\gamma\neq 0$ is a scaling factor to be determined.

Then let 
$$\ket{a_x} = \bigoplus_{j, q} \ket{a_{x, jq}}, \qquad \quad \ket{w_x} = \bigoplus_{j, q} \ket{w_{x, jq}}.$$
Finally let $A$ be a matrix whose $x$-th row is $\bra{a_x}$. The set of columns of this matrix has a natural partition into $n\times \ell$  subsets which results in sets $I_{j, q}$ of our canonical span program (see Figure~\ref{fig:A-alpha}). 

\begin{figure}[t!]
\centering \includegraphics[scale=.93]{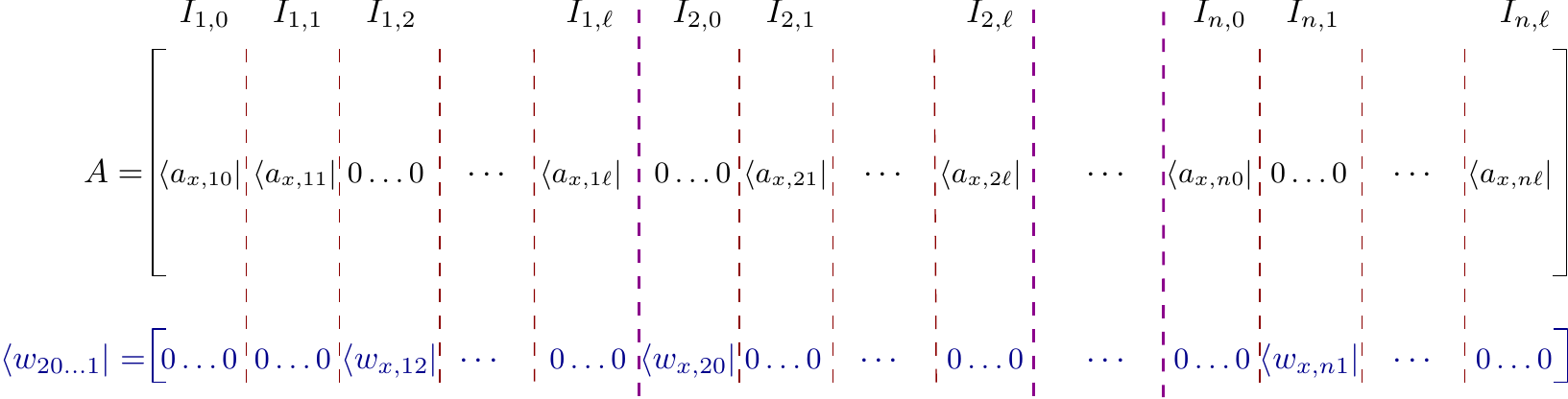}
\caption{The matrix $A$ and one of its rows corresponds to the input $x=1,0,\ldots, 1$ and its positive witness $\ket{w_{20\ldots 1}}$ in a canonical span program with orthogonal inputs. }
\label{fig:A-alpha}
\end{figure}

We now verify that these define a valid span program evaluating $f$. 
First of all, from the definition of $A$, the negative witness $\ket{\bar w_x} = \ket{e_x}$ is orthogonal to all available vectors.
Second, it is clear that in $\ket{w_x}$ the coordinates associated to unavailable vectors are zero. Moreover, we have
$$\bra{e_y}A\ket{w_x} =  \sum_{j, q} \bra{a_{y, jq}} w_{x, jq}\rangle = \sum_{j: x_j\neq y_j} \bra{u_{y, j}} v_{x, j}\rangle =1-\delta_{f(x), f(y)}.$$
Therefore, if $f(x)=\alpha$ we have
$$A\ket {w_x} = \sum_{y: f(y)\neq f(x)}\ket{e_y} = \ket{t_\alpha},$$
as required. Thus the matrix $A$ and $\ket{w_x}$'s form a valid canonical span program for $f$. We then compute its complexity:
$$ \big\|  \ket{w_x}\big\| ^2=\sum_{j, q}   \big\|\ket{w_{x, jq}}\big\|^2= \sum_{j}   \big\|\ket{w_{x, jx_j}}\big\|^2   = \gamma^2\sum_{j}   \big\|  |v_{xj}\rangle\big\| ^2.$$
We also have
\begin{align*}
\big\|  A^\dagger\ket{\bar{w}_x}\big\| ^2 &=\big\|  A^\dagger|e_x\rangle\big\| ^2 
=\sum_{j,q} \big\| \ket{a_{x, jq}}\big\|^2 \\
&= \gamma^{-2}\sum_{j,q\neq x_j} \big\| \ket{u_{x, j}}\big\|^2  
=(\ell-1) \gamma^{-2 }\sum_{j} \big\| \ket{u_{x, j}}\big\|^2.
\end{align*}
Therefore,
\begin{align*}
\mathrm{wsize}(P,w,\bar{w}) & = \max_x \max\big\{\big\|  \ket{w_x}\big \| ^2,\big\|  A^\dagger\ket{\bar{w}_x}\big\| ^2\big\} \\
& =\max_x\max\bigg\{ \gamma^{2} \sum_{j} \big\| |v_{xj}\rangle\big\| ^2, (\ell-1) \gamma^{-2}  \sum_j   \big\| |u_{xj}\rangle\big \| ^2\bigg\}.
\end{align*}
Letting $\gamma^2 = \sqrt{\ell-1}$ we find that $\mathrm{wsize}(P,w,\bar{w})$ is $\sqrt{\ell-1}$ times the objective value of SDP~\eqref{eq:dual-SDP} for vectors $\ket{v_{x, j}}$'s and $\ket{u_{x, j}}$'s.

\end{document}